\documentclass{article}
\usepackage{amsmath}
\usepackage{amsthm}
\usepackage{amssymb}
\usepackage{mathabx}   

\usepackage{bussproofs}

\usepackage{enumitem}  
\usepackage{multicol}

\usepackage{booktabs}

\usepackage{tikz}
\usetikzlibrary{decorations.markings}
\tikzset{negated/.style={
		decoration={markings,
			mark= at position 0.5 with {
				\node[transform shape] (tempnode) {$\backslash$};
			}
		},
		postaction={decorate}
	}
}

\def\texorpdfstring#1{}

\newcommand{\logspace}{{\mathrm{L}}}
\newcommand{\NL}{{\mathrm{NL}}}
\newcommand{\coNL}{{\mathrm{coNL}}}
\newcommand{\ACzero}{{\mathrm{AC}^0}}
\newcommand{\Alogtime}{{\mathrm{ALogTime}}}
\newcommand{\NC}{{\mathrm{NC}}}
\newcommand{\Ptime}{{\mathrm{P}}}
\newcommand{\NP}{{\mathrm{NP}}}
\newcommand{\coNP}{{\mathrm{coNP}}}
\newcommand{\IDelta}{{\mathrm{I}\Delta}}
\newcommand{\PV}{{\mathrm{PV}}}
\newcommand{\VL}{{\mathrm{VL}}}
\newcommand{\VNL}{{\mathrm{VNL}}}
\newcommand{\VNC}{{\mathrm{VNC}}}
\newcommand{\PSA}{{\mathrm{PSA}}}
\newcommand{\Gtheory}{{\mathrm G}}
\newcommand{\Stheory}{{\mathrm S}}
\newcommand{\Ttheory}{{\mathrm T}}
\newcommand{\Utheory}{{\mathrm U}}
\newcommand{\Sonetwo}{{\Stheory^1_2}}
\newcommand{\eF}{{\mathrm e \cal F}}
\newcommand{\CF}{{\mathrm{CF}}}
\newcommand{\GL}{{\mathrm{GL}}}
\newcommand{\GNL}{{\mathrm{GNL}}}
\newcommand{\SigmaCNF}{{\Sigma\mathrm{CNF}}}
\newcommand{\SigmaKrom}{{\Sigma\mathrm{Krom}}}
\newcommand{\LDT}{{\mathrm{LDT}}}
\newcommand{\LNDT}{{\mathrm{LNDT}}}
\newcommand{\eLDT}{{\mathrm{eLDT}}}
\newcommand{\eLNDT}{{\mathrm{eLNDT}}}
\newcommand{\LK}{{\mathrm{LK}}}
\newcommand{\eLK}{{\mathrm{eLK}}}
\newcommand{\dLK}[1]{\hbox{\rm #1-}{\mathrm{LK}}}
\newcommand{\oneLK}{{\dLK 1}}
\newcommand{\twoLK}{{\dLK 2}}
\newcommand{\DT}{{\mathrm{DT}}}
\newcommand{\NDT}{{\mathrm{NDT}}}
\newcommand{\eDT}{{\mathrm{eDT}}}
\newcommand{\eNDT}{{\mathrm{eNDT}}}
\newcommand{\TreeSys}[1]{{\hbox{\rm Tree-}#1}}
\newcommand{\TreeLDT}{{\TreeSys\LDT}}
\newcommand{\TreeLNDT}{{\TreeSys\LNDT}}
\newcommand{\TreeLK}{{\TreeSys\LK}}
\newcommand{\TreeoneLK}{{\TreeSys\oneLK}}
\newcommand{\TreetwoLK}{{\TreeSys\twoLK}}
\newcommand{\TreeELK}{{\TreeSys\eLK}}
\newcommand{\TreeDLK}[1]{{\TreeSys{\dLK{#1}}}}
\newcommand{\TreeELDT}{{\TreeSys{\eLDT}}}

\newcommand{\Conj}{{\mathrm{Conj}}}
\newcommand{\Disj}{{\mathrm{Disj}}}
\newcommand{\Tms}{{\mathrm{Tms}}}
\newcommand{\Cls}{{\mathrm{Cls}}}
\newcommand{\DTms}{{\mathrm{DTms}}}

\newcommand{\NF}{{\mathrm{NF}}}
\newcommand{\AndDT}{{\hbox{\sc And}}}
\newcommand{\AndNDT}{\AndDT}
\newcommand{\OrDT}{{\hbox{\sc Or}}}

\newcommand{\DTx}{{\mbox{\sc Dt}}}
\newcommand{\Reach}{{\mbox{\it Reach}}}

\def\sequent{{\mbox{\Large $\,\rightarrow\,$}}}
\def\fCenter{\sequent}
\def\liff{\leftrightarrow}
\def\pprime{{\prime\prime}}
\newcommand{\taller}{\raise 0.8ex \hbox{\vphantom{)}}}

\def\srt{{\smalltriangleright}}
\def\veceprime{{\vec e ^{\, \prime}}}
\def\vecepprime{{\vec e ^{\, \pprime}}}
\def\vecpprime{{\vec p ^{\, \prime}}}
\def\vecppprime{{\vec p ^{\, \pprime}}}

\newcommand{\rddots}{{\mathinner{\mkern1mu\raise1pt\vbox{\kern7pt\hbox{.}}\mkern2mu\raise4pt\hbox{.}\mkern2mu\raise7pt\hbox{.}\mkern1mu}}}
\newcommand{\proofdots}{{\ddots\vdots\,\rddots}}
\newcommand{\proofdotsL}[1]{{\proofdots\hbox to0pt{\kern10pt\raisebox{1ex}{$#1$}\hskip 0pt minus 2in}}}

\def\Largelongleftrightarrow{{\mbox{\Large $\longleftrightarrow$}}}

\def\lorLeft{{\mbox{\it $\lor$\kern-1pt -l}}}
\def\lorRight{{\mbox{\it $\lor$\kern-1pt -r}}}
\def\landLeft{{\mbox{\it $\land$\kern-1pt -l}}}
\def\landRight{{\mbox{\it $\land$\kern-1pt -r}}}

\newtheorem{theorem}{Theorem}[section]
\newtheorem{lemma}[theorem]{Lemma}
\newtheorem{proposition}[theorem]{Proposition}

\newtheorem{claim}[theorem]{Claim}
\newtheorem{question}[theorem]{Question}

\theoremstyle{definition}
\newtheorem{definition}[theorem]{Definition}
\newtheorem{example}[theorem]{Example}
\newtheorem{remark}[theorem]{Remark}

\newlist{integerenumerate}{enumerate}{1}
\setlist[integerenumerate]{label=\textnormal{(\arabic*)}, ref=(\arabic*)}
\newlist{alphenumerate}{enumerate}{1}
\setlist[alphenumerate]{label=\textnormal{(\alph*)}, ref=(\alph*)}

\title{Proof complexity of systems of (non-deterministic) decision trees and branching programs}
\author{%
Sam Buss\thanks{Supported in part by Simons Foundation grant 578919} \\[-0.6ex]
{\small Dept.\ of Mathematics } \\[-0.8ex]
{\small UC San Diego} \\[-0.8ex]
{\small \tt sbuss@ucsd.edu}
\and
Anupam Das\thanks{Supported by a a Marie Sk\l{}odowska-Curie fellowship,
\emph{Monotonicity in Logic and Complexity}, ERC project 753431.} \\[-0.6ex]
{\small Dept.\ of Computer Science} \\[-0.8ex]
{\small University of Copenhagen } \\[-0.8ex]
{\small \tt anupam.das@di.ku.dk}
\and
Alexander Knop \\[-0.6ex]
{\small Dept.\ of Mathematics } \\[-0.8ex]
{\small UC San Diego} \\[-0.8ex]
{\small \tt aknop@ucsd.edu}}

\begin{document}
\maketitle

\begin{abstract}
This paper studies propositional proof systems in which
lines are sequents of decision trees or branching programs ---
deterministic and nondeterministic.
The systems LDT and LNDT
are propositional proof systems in which lines represent
deterministic or non-deterministic decision trees.
Branching programs are modeled as decision dags.
Adding extension to LDT and LNDT gives systems
eLDT and eLNDT in which lines represent
deterministic and non-deterministic branching programs,
respectively.

Deterministic and non-deterministic
branching programs correspond to log-space (L) and nondeterministic
log-space (NL).
Thus the systems eLDT and eLNDT are propositional
proof systems that reason with (nonuniform)
L and NL properties.

The main results of the paper are simulation and
non-simulation results for tree-like and dag-like
proofs in the systems LDT, LNDT, eLDT, and eLNDT.
These systems are also compared with Frege
systems, constant-depth Frege systems and extended
Frege systems.  
\end{abstract}

\section{Introduction}

Propositional proof systems are widely studied because of their
connections to complexity classes and their usefulness
for computer-based reasoning. The first connections to computational
complexity arose largely from the work of
Cook and Reckhow \cite{Cook:PV,CookReckhow:proofsstoc,CookReckhow:proofs},
showing a connection to the $\NP$-$\coNP$ question.
These results, building on the work of Tseitin~\cite{Tseitin:derivation}
initiated the study of the relative efficiency of propositional
proof systems. The present paper introduces
propositional proof systems that are closely connected to
log-space ($\logspace$) and nondeterministic log-space ($\NL$).

Our original motivation for this study was to investigate propositional
proof systems corresponding to the first-order bounded
arithmetic theories $\VL$ and $\VNL$ for $\logspace$ and $\NL$,
see~\cite{CookNguyen:book}.
This follows a long line of work defining formal theories of bounded
arithmetic that correspond to computational complexity
classes, as well as to provability in propositional proof
systems. The first results of this type were due (independently)
to Paris and Wilkie~\cite{ParisWilkie:counting}
who gave a translation from $\IDelta_0$ to
constant-depth Frege ($\ACzero$-Frege)
proofs and to Cook~\cite{Cook:PV}
who gave a translation from $\PV$ to
extended Frege ($\eF$) proofs. Since the first-order
bounded arithmetic theory~$\Sonetwo$ is conservative
over the equational theory~$\PV$,
Cook's translation also applies
to the bounded arithmetic theory $S^1_2$~\cite{Buss:bookBA}.
As shown in the table below,
similar propositional translations have since been given
for a range of other theories, including first-order, second-order
and equational theories.
\begin{center}
    \begin{tabular}{c c c l}
        \toprule
        Formal & Propositional & Complexity & \\
        Theories & Proof Systems & Class & \\
        \midrule
        $\PV$, $\Sonetwo$ & $\eF$ & $\Ptime$ & \cite{Cook:PV,Buss:bookBA} \taller \\[0.3ex]
        $\PSA$, $\Utheory^1_2$ & QBF & PSPACE & \cite{Dowd:PSA,Buss:bookBA} \\[0.3ex]
        $\Ttheory^i_2$, $\Stheory^{i+1}_2$ & $\Gtheory_i$, $\Gtheory_{i+1}^*$
              & $\Ptime^{\Sigma_i^{\mathrm p}}$
              & \cite{KrajicekPudlak:quantified,KrajicekTakeuti:indfreeprove,Buss:bookBA}
              \\[0.3ex]
        $\VNC^0$ & Frege ($\cal F$) & $\Alogtime$
              & \cite{CookMorioka:NCone,CookNguyen:book,Arai:AID} \\[0.3ex]
        $\VL$ & $\GL^*$ & $\logspace$ & \cite{Perron:logspace,CookNguyen:book} \\[0.3em]
        $\VNL$ & $\GNL^*$ & $\NL$ & \cite{Perron:thesis,CookNguyen:book} \\
        \bottomrule
    \end{tabular}
\end{center}
The first three theories are first-order theories;
the last three theories are second-order. The last three
theories could also be viewed as
multi-sorted first-order theories, but their formalization as second-order
theories makes it possible for them to work elegantly with
weak complexity classes.
(For an introduction to these and related results, see the books
\cite{Buss:bookBA,CookNguyen:book,Krajicek:book,Krajicek:bookTwo}.)

A hallmark of the propositional translations in the table
above is that the lines in the propositional proofs express
(nonuniform) properties in the corresponding complexity
class. For instance, a line in a Frege proof is a
propositional formula, and the evaluation problem
for propositional formulas is complete for alternating
log-time ($\Alogtime$), cf.~\cite{Buss:BFVP}. Likewise, a line in a
$\eF$ proof is (implicitly) a Boolean circuit,
and the Boolean circuit value problem is well known to be complete
for $\Ptime$, cf.~\cite{Ladner:CirVal}.
In the usual formulation of $\eF$,
the lines only ``implicitly''
express Boolean circuits, since it is necessary to expand
the
definitions of extension variables to form the circuit;
however, Je\v{r}\'{a}bek~\cite{Jerabek:dualweakphp}
made this connection explicit in a propositional proof system
Circuit-Frege $\CF$, in which lines are actually Boolean circuits.

The present paper's main goal is to define alternatives
for the
proof systems $\GL^*$ and $\GNL^*$ corresponding to log-space
and nondeterministic log-space, see \cite{Perron:logspace,Perron:thesis,CookKolokolova:Gradel,CookKolokolova:NL}.  The proof system $\GL^*$
restricts cut formulas to be ``$\SigmaCNF(2)$'' formulas;
the subformula property then implies that proofs contain only
$\SigmaCNF(2)$ formulas when proving $\SigmaCNF(2)$ theorems.
$\GNL^*$ similarly restricts cut formulas to be ``$\SigmaKrom$'' formulas.
(A $\SigmaKrom$ formula has the form
$\exists \vec{z} \varphi(\vec{z}, \vec{x})$, where $\varphi$ is a conjunction
$C_1 \land C_2 \land \cdots \land C_n$ with
each $C_i$ a disjunction of any number of $x$-literals and at most two
$z$-literals.)
$\SigmaCNF(2)$ and $\SigmaKrom$
have expressive power equivalent to
nonuniform $\logspace$ and $\NL$ respectively~\cite{Johannsen:detLogSpace,Gradel:SOfragments},
but they are
are somewhat ad hoc classes of quantified
formulas, and their connections to $\logspace$ and $\NL$
are indirect.
In this paper, we propose new proof systems,
called $\eLDT$ and $\eLNDT$, intended to be alternatives for $\GL^*$ and
$\GNL^*$ respectively. The lines in $\eLDT$ and $\eLNDT$ proofs are
sequents of formulas expressing \emph{branching programs}
and \emph{nondeterministic branching programs}, respectively.
This follows an earlier unpublished suggestion of
S.\ Cook~\cite{Cook:EdinburghSlides}, who gave a system for~$\logspace$ based on branching programs via ``Prover-Liar'' games (see \cite{BussPudlak:howtolie}).
The advantage of our systems is that deterministic and nondeterministic branching
programs correspond directly to nonuniform $\logspace$ and $\NL$ respectively
and do not require the use of quantified formulas. (See \cite{Wegener00}
for a comprehensive introduction to branching programs.)

To design the proof systems $\eLDT$ and $\eLNDT$,
we need to choose representations for
branching programs. For this, we use a formula-based
representation, as this fits well into the customary
frameworks for proof systems. The formulas appearing in
$\eLDT$ and $\eLNDT$ proofs will be descriptions of
\emph{decision trees}.  Decision trees are not as
powerful as branching programs since branching programs
may be dags instead of trees. Accordingly, we also
allow extension variables. The use of extension variables
allows decision trees to express branching programs; this is similar
to the way the extension variables in extended Frege
proofs allow formulas to express circuits.
An example is given in the figure on page~\pageref{fig:bp-to-edt}.

We start in Section~\ref{sec:LDT_LNDT} describing
proof systems $\LDT$ and $\LNDT$ that work with
just deterministic and nondeterministic
decision trees (without extension variables). Deterministic
decision trees are represented by formulas
using a single
``case'' or ``if-then-else'' connective, written in
infix notation $A p B$, which means
``if $p$ is false, then~$A$, else~$B$''.  The condition~$p$
is required to be a literal, but $A$ and $B$ are
arbitrary formulas. The system
$\LDT$ is a sequent calculus system in which all formulas
are decision trees. Nondeterministic decision trees are represented
with formulas that may also use disjunctions, allowing
formulas of the form $A\lor B$. The system
$\LNDT$ is a sequent calculus in which all
formulas are nondeterministic decision trees.

$\LDT$ and $\LNDT$ are weak systems; in fact, they are both
polynomially simulated by depth-2 $\LK$
(the sequent calculus~$\LK$ with all formulas of depth two).
Figure~\ref{fig:simulations} shows the equivalences between
systems as currently established. The equivalences and
separations that concern $\LDT$ and $\LNDT$ are proved in
Section~\ref{sec:NDTsystems}.

Section~\ref{sec:eLDT_eLNDT} introduces
the proof systems $\eLDT$ and $\eLNDT$ for branching
programs and nondeterministic branching programs.
These again are sequent calculus systems. These systems
are obtained from $\LDT$ and $\LNDT$ by adding the
extension rule, thereby effectively changing the
expressive power of formulas from decision trees
to decision diagrams.  (Decision diagrams are of course
the same as a branching programs).

An important
issue is designing these proof systems is how to
handle isomorphic or bisimilar branching programs. Two branching
programs $A$ and~$B$ are isomorphic if there
is an isomorphism (a bijection) between the nodes of the branching
programs.  The most convenient solution perhaps would be to allow
the propositional proof systems to freely replace any branching
program with any isomorphic branching program: for this,
we could allow ``isomorphism axioms'' or ``bisimilarity axioms''
$A \liff B$ whenever the two programs are isomorphic or bisimilar
(respectively). For instance, isomorphism axioms of this type
were used by Je\v{r}\'{a}bek~\cite{Jerabek:dualweakphp} for the reformulation
of extended Frege using Boolean circuits as lines. The problem with using
 isomorphism or bisimilarity axioms is that --- as argued in the next
 paragraph --- the isomorphism and bisimilarity problems for branching
programs are known to be in $\NL$, but they not known to be
in~$\logspace$. In other words, it is open whether valid isomorphism
or bisimilarity axioms are recognizable in log-space. This make the use
of these axioms undesirable, at least for $\eLDT$, as it is
a proof system for log-space.

As a sketch of how to recognize bisimilarity with a
$\NL$ algorithm, let $A$ and $B$ be branching programs.
A ``path'' in either $A$ or $B$ is specified by some sequence
of values $v_1, v_2, v_3,\dots$ of \emph{true} or \emph{false}
(1 or 0): a path is traversed in the obvious way, starting the
source of the branching program, and using the value $v_i$ to decide
how to branch when reaching the $i$-th vertex. (Note this allows a
variable to be given conflicting truth values at different points in
the path.) Then $A$ and $B$ are \emph{bisimilar} provided that
any given path in~$A$ reaches a vertex labelled with
a literal~$p$ or a sink vertex labelled with 1 or~0
if and only if the same path in~$B$ reaches a vertex labelled
with the same literal~$p$ or a sink vertex labelled with the
same value $1$ or~$0$. This is clearly
$\coNL$ verifiable; namely, co-nondeterministically choose
a path to traverse simultaneously in $A$ and~$B$.  Two branching
programs are \emph{isomorphic} provided that they are bisimilar, and
that in addition, any two paths reach distinct nodes in~$A$ if and only if
they reach distinct nodes in~$B$.  This property clearly can
also be checked co-nondeterministically. Since $\NL = \coNL$ (cf.~\cite{Immerman:nspace,Szelepcsenyi:NLCSL}),
these properties are also in~$\NL$.

One way to handle isomorphism and bisimilarity
would be to nonetheless use (say) isomorphism axioms,
but require that they be accompanied by an explicit isomorphism.
In our setting, this might mean giving an explicit
renaming of extension variables that makes the two formulas and the
definitions of their associated extension variables identical.
We instead adopt a more conservative approach, and do not
allow isomorphism axioms. Instead, the equivalence of
isomorphic branching programs (and more generally, of
bisimilar branching programs) is proved explicitly, using
induction on the size of the branching programs.

Since formulas in $\eLDT$ and $\eLNDT$ proofs express nonuniform
$\logspace$ and $\NL$ properties, respectively, they are intermediate
in expressive power between Boolean formulas (expressing $\NC^1$
properties) and Boolean circuits (expressing nonuniform $\Ptime$
properties).
Thus it is not surprising that, as shown in Figure~\ref{fig:simulations},
these two systems are between Frege and extended Frege in strength.
In addition, since $\NL$ properties can be expressed by quasipolynomial
formulas, it is not unexpected that Frege proofs can quasipolynomially
simulate $\eLNDT$, and hence $\eLDT$.  These facts are proved
in Section~\ref{sec:eDTequivs}.

\begin{figure}[!t]
    \begin{center}
    	\begin{tikzpicture}[>=latex, thick]
	        \node (T1F) [draw] at (0, 1.2) {$\TreeDLK{1}$};
	        \node (TLDT) [draw] at (0, 3) {$\TreeLDT$};
	        \node (LDT) [draw] at (0, 4.5) {$\TreeDLK{2} \underset{\text{\tiny Thm~\ref{thm:2LK_LNDT}}}{\Largelongleftrightarrow} \TreeLNDT \underset{\text{\tiny Thm~\ref{thm:tree-lndt-equiv-ldt}}}{\Largelongleftrightarrow} \LDT \overset{\text{\tiny Thm~\ref{thm:ldt-sim-1lk}}}{\underset{\text{\tiny Thm~\ref{thm:1lk-sim-ldt}}}{\Largelongleftrightarrow}} \oneLK$
	        };
	        \node (LNDT) [draw] at (0, 6) { $2$-$\LK \underset{\text{\tiny Thm~\ref{thm:2LK_LNDT}}}{\Largelongleftrightarrow} \LNDT$ };
	        \node (F) [draw] at (0, 7.5) { Frege $\Largelongleftrightarrow \LK \Largelongleftrightarrow \TreeLK$};
	        \node (eLDT) [draw] at (0, 9) {$\eLDT$ };
	        \node (eLNDT) [draw] at (0, 10) {$\eLNDT$ };
	        \node (eF) [draw] at (0, 11) {$\eLK \Largelongleftrightarrow \TreeELK$};
	        \draw[->] (TLDT) to[out=-60, in=60] node[right] {\tiny Thm~\ref{thm:ldt-sim-1lk}}  (T1F);
	        \draw[->] (T1F) to[out=120, in=-120] node[left]
	                    {$\underset{\text{\tiny Thm~\ref{thm:1lk-sim-ldt}}}{\text{~~~qp}}$}  (TLDT);
	        \draw[->] (LDT) to[out=-60, in=60] (TLDT);
	        \draw[->, dashed] (TLDT) to[out=120, in=-120] (LDT);
	        \draw[->] (LNDT) to[out=-60, in=60] (LDT);
	        \draw[->, dashed] (LDT) to[out=120, in=-120] (LNDT);
	        \draw[->] (F) to[out=-60, in=60] (LNDT);
	        \draw[->, dashed] (LNDT) to[out=120, in=-120] (F);
	        \draw[->] (eLDT) to node[right]{\tiny Thm~\ref{thm:ELDT_LK}} (F);
	        \draw[->] (eLNDT) to (eLDT);
	        \draw[->] (eF) to (eLNDT);
	        \draw[->] (F) to[out=150, in=-150] node[left]
	                    {$\underset{\text{\tiny Thm~\ref{thm:LK_ELNDT}}}{\text{~~~qp}}$}  (eLNDT);
	    \end{tikzpicture}
    \end{center}

    \caption{Relations between proof systems.
        $\rightarrow$ means ``polynomially simulates''; $\rightarrow_{qp}$
        means ``quasipolynomially simulates'';  $\dashrightarrow$ means
        ``exponentially separated from''. $d$-$\LK$ is the system of dag-like
        $\LK$ proofs with only depth $d$ formulae occurring (atomic formulae
        have depth $0$)
        By default,
        all proof systems
        allow dag-like proofs, unless they are labeled as ``Tree''.
    }
    \label{fig:simulations}
\end{figure}
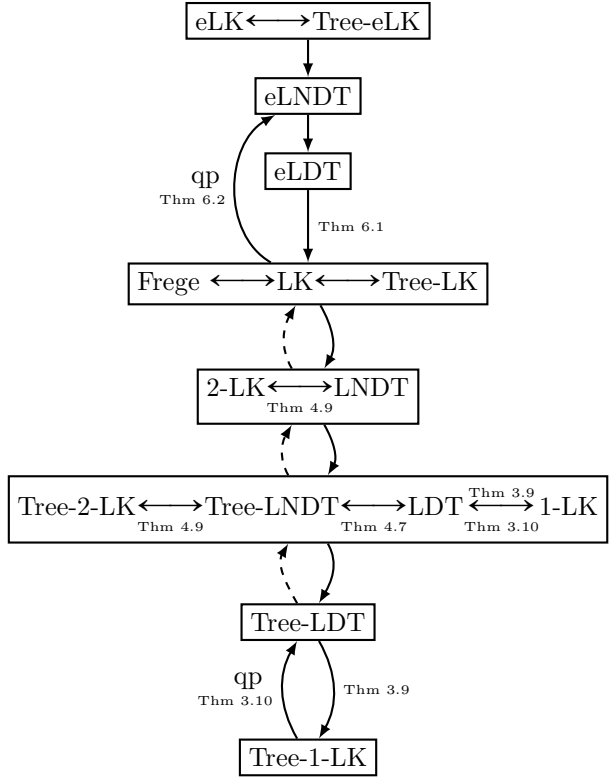

\section{
Decision tree formulas and \texorpdfstring{$\LDT$}{LDT} proofs
}
\label{sec:LDT_LNDT}
\label{sec:DTdefs}

This section describes decision tree ($\DT$) formulas, and the associated
sequent calculus proof system $\LDT$.  All our proof systems are
\emph{propositional} proof systems with variables $x, y, z \dots$ intended to
range over the Boolean values \emph{False} and \emph{True}.  We use $0$ and $1$
to denote the constants \emph{False} and \emph{True}, respectively.
A \emph{literal} is either a propositional variable $x$ or a negated
propositional variable $\overline x$. We use use variables $p,q,r,\dots$ to
range over literals.

The only connective for forming decision tree  formulas ($\DT$ formulas)
is the 3-ary ``case'' function, written in infix notation
as $(A p B)$ where $A$ and~$B$ are formulas and $p$ is
required to be a literal. This informally means
``if $p$ is false, then $A$, else $B$''. The syntax is formalized
by:
\begin{definition}\label{def:DT}
The \emph{decision tree formulas}, or \emph{$\DT$ formulas}
for short, are inductively defined by
\begin{integerenumerate}
\item any literal $p$ is a $\DT$ formula, and
\item if $A$ and $B$ are $\DT$ formulas
and $p$ is a literal, then $(A p B)$ is a $\DT$ formula.
We call $p$ a \emph{decision literal}.
\end{integerenumerate}
\end{definition}
The parentheses in~(2) ensure unique
readability, but we informally write just $ApB$ when the meaning is
clear.

Suppose $\alpha$ is a truth assignment to the variables; the
semantics of $\DT$ formulas is defined by extending $\alpha$ to
be a truth assignment to all $\DT$ formulas by inductively
defining
\begin{eqnarray}
\label{eq:TruthDefnDT}
\alpha(\overline x) & = & 1-\alpha(x) \\[1ex]
\nonumber
\alpha( ApB ) &=&
\begin{cases}
    \alpha(A) \quad & \text{if } \alpha(p) = 0 \\
    \alpha(B)       & \text{otherwise.}
\end{cases}
\end{eqnarray}
It is important that only \emph{literals} $p$
may serve as the decision literals in DT formulas.
Notably, for
$C$ a complex formula, an expression of the form
$( A \, C \, B)$, which evaluates to $A$ if $C$~is
true and to $B$ if $C$~is false, would in general be only a
decision \emph{diagram}, not a decision tree.

Although there is no explicit negation
of $\DT$ formulas, we informally define the negation $\overline A$ of a $\DT$ formula inductively
by letting $\overline {\overline x}$ denote~$x$,
and letting $\overline {A p B}$ denote the formula
$\overline A\, p\, \overline B$.
Of course $\overline A$ is a $\DT$ formula whenever $A$ is,
and $\overline A$ correctly expresses the negation of~$A$.
Notice also that negative decision literals are `syntactic sugar',
since $A \bar{p} B$ is equivalent to $B p A$.
Nonetheless the notation is useful for making later definitions more intuitive.

Our definition of $\DT$ formulas is somewhat different from
the usual definition of decision trees. The more common
definition would allow $0$ and $1$ as atomic formulas
instead of literals~$p$
as in condition~(1) of Definition~\ref{def:DT}.
We call such formulas $0/1$-$\DT$ formulas.
$\DT$ formulas and $0/1$-$\DT$ formulas are
are equivalent in expressive power. The constants
$0$ and~$1$ are equivalent to $p p \overline p$ and
$\overline p p p$, for any literal~$p$.
More generally, any formula $0pA$, $1pA$,
$Ap0$ or $Ap1$ is equivalent to
$p p A$, $\overline p p A$, $A p \overline p$, or $A p p$,
respectively. Conversely,
a literal~$p$, when used as atom, is equivalent to $0p1$.

\begin{remark}
[Expressive power of decision trees]
It is easy to decide the validity or satisfiability
of a $\DT$ formula with a log-space algorithm.
A $\DT$ formula is presented as fully parenthesized,
syntactically correct formula, and it is well-known that
formulas can be efficiently parsed in $\logspace$.
To check satisfiability, for example, one examines
each leaf in the formula tree (each atomic subformula~$p$)
and verifies whether the
path from the root to the leaf, assigning true to the
literal at the leaf, is permitted under
any consistent assignment of truth values to
variables.

The size of a $\DT$ formula~$A$ is the number of occurrences of
atomic formulas in~$A$.  Recall that a (Boolean) CNF
formula is a conjunction of disjunctions of literals;
each such disjunction is called a \emph{clause}. Likewise
a (Boolean) DNF formula is a disjunction of conjunctions
of literals; each such conjunction is called a \emph{term}.
A $\DT$ formula~$A$ of size $n$ can be expressed
as a DNF formula of size $O(n^2)$ with at most $n$ disjuncts.
This is defined formally
as $\Tms(A)$ in Section~\ref{sec:LK_and_LDT}: informally,
$\Tms(A)$ is formed by converting the formula to
a $0/1$-$\DT$ formula, and then forming the disjunction,
taken over all leaves labelled by a~$1$, of the terms
expressing that that leaf is reached.  A dual construction expresses
a $\DT$ formula~$A$ as a CNF, denoted $\Cls(A)$
of size $O(n^2)$ with at most $n$ conjuncts.

It is folklore that the construction can be partially
reversed: namely any Boolean function that is equivalently
expressed by a DNF $\varphi$ and a CNF $\psi$ can be
represented by a $\DT$ formula of size quasipolynomial
in the sizes of $\varphi$ and $\psi$. This bound
is optimal, as \cite{JuknaRSW99} proves a quasipolynomial
lower bound.
\end{remark}

We next define the proof system $\LDT$ for reasoning
about $\DT$ formulas.  Lines in an $\LDT$ proof are sequents,
hence they express disjunctions of $\DT$'s.  Thus lines
in $\LDT$ proofs can express DNF properties: for these,
the validity problem is non-trivial, in fact, $\coNP$-complete.

\begin{definition}
A \emph{cedent}, denoted $\Gamma$, $\Delta$ etc.,
is a multiset of formulas; we often use
commas for multiset union, and write $\Gamma, A$
for the multiset $\Gamma,\{A\}$. A \emph{sequent}
is an expression $\Gamma\sequent \Delta$ where
$\Gamma$ and~$\Delta$ are cedents. $\Gamma$ and $\Delta$
are called the \emph{antecedent} and \emph{succedent},
respectively.
\end{definition}
The intended meaning
of $\Gamma\sequent \Delta$
is that if every formula in $\Gamma$ is true, then some
formula in~$\Delta$ is true.  Accordingly, $\Gamma \sequent \Delta$
is true under a truth assignment~$\alpha$ iff
$\alpha(A) = 0$ for some $A \in \Gamma$ or
$\alpha(A) = 1$ for some $A \in \Delta$. A sequent is
\emph{valid} iff it is true for every truth assignment.

\begin{definition}\label{def:LDT}
The sequent calculus $\LDT$ is a proof system
in which lines are sequents of $\DT$ formulas.
The valid initial sequents (axioms) are, for $p$~any
literal,
\[
p \sequent p\qquad\qquad
p, \overline p \sequent \qquad\qquad
\sequent p, \overline p .
\]

\noindent
The rules of inference are:
\bigskip

\noindent
\def\myExtraSpace{\hspace{4em}}
\begin{tabular}{lc@{\myExtraSpace}c}
{\bf Contraction rules:} &
\Axiom$A , A , \Gamma \fCenter \Delta$
\LeftLabel{\emph{c-l:}}
\UnaryInf$A, \Gamma \fCenter \Delta$
\DisplayProof &
\Axiom$\Gamma \fCenter \Delta, A , A$
\LeftLabel{\emph{c-r:}}
\UnaryInf$\Gamma \fCenter \Delta, A$
\DisplayProof \\[4ex]
{\bf Weakening rules:} &
  \Axiom$\Gamma \fCenter \Delta$
  \LeftLabel{\emph{w-l:}}
  \UnaryInf$A, \Gamma \fCenter \Delta$
  \DisplayProof &
  \Axiom$\Gamma \fCenter \Delta$
  \LeftLabel{\emph{w-r:}}
  \UnaryInf$\Gamma \fCenter \Delta, A$
  \DisplayProof \\[4ex]
{\bf Cut rule: } & \multicolumn{2}{c}{
\Axiom$\Gamma \fCenter\Delta, A$
\Axiom$A, \Gamma \fCenter \Delta$
\LeftLabel{\emph{cut:}}
\BinaryInf$\Gamma \fCenter \Delta$
\DisplayProof } \\[4ex]
{\bf Decision rules:} & \multicolumn{2}{c}{
\Axiom$A,\Gamma \fCenter \Delta, p$
\Axiom$p, B, \Gamma \fCenter \Delta$
\LeftLabel{\emph{dec-l:}}
\BinaryInf$ApB, \Gamma \fCenter \Delta$
\DisplayProof } \\[4ex]
& \multicolumn{2}{c}{
\Axiom$\Gamma \fCenter \Delta, A, p$
\Axiom$p, \Gamma \fCenter \Delta, B$
\LeftLabel{\emph{dec-r:}}
\BinaryInf$\Gamma \fCenter \Delta, ApB$
\DisplayProof }
\end{tabular}
\bigskip

\noindent
Proofs are, by default, dag-like.
I.e.\ a \emph{proof} of a sequent $S$ in $\LDT$ is a sequence
$(S_0, \dots , S_n)$ such that $S$ is $S_n$ and each $S_k$ is either an
initial sequent or is the conclusion of an inference step whose premises
occur amongst $(S_i)_{i<k}$. The subsystem where proofs are restricted
to be tree-like (i.e.\ trees of sequents composed by inference steps)
is denoted $\TreeLDT$.

The \emph{size} of a proof is the sum of the sizes of the formulas
occurring in the proof.
\end{definition}
The inference rules that are new to $\LDT$ are the two
decision rules, \emph{dec-l} and \emph{dec-r}.  Since
$ApB$ is equivalent to
$(A \lor p) \land (B \lor \overline p)$,
the lower sequent of a \emph{dec-r} is true (under some
fixed truth assignment) iff both upper sequents are true
under the same assignment. This property of \emph{dec-r}
inferences is
called ``invertibility''; in particular, it means
that the \emph{dec-r} rule is sound.
Similarly, since $ApB$ is also equivalent to
$(A\land \overline p)\lor (B \land p)$, the
\emph{dec-l} rule is also sound and invertible.

\begin{remark}
[Cut-free completeness]
\label{rem:cut-free-completeness}
The invertibility properties also imply that the
cut-free fragment of $\LDT$ is complete.  To prove this
by induction on the complexity of sequents,
start with a valid sequent $\Gamma \sequent \Delta$;
choose any non-atomic formula $ApB$ in $\Gamma$ or $\Delta$,
and apply the appropriate decision rule \emph{dec-l} or
\emph{dec-r} that introduces this formula.
The upper sequents of this inference are also valid. Since
they have logical complexity strictly less then the logical
complexity of $\Gamma\sequent\Delta$, and thus, arguing
by induction, they have cut-free proofs. The base case
of the induction is when $\Gamma\sequent\Delta$ contains
only atomic formulas; in this case, it can be inferred from an initial
sequent with weakenings.
Note that this shows in fact, that any valid
sequent can be proved in $\LDT$ using only
decision rules, weakenings, and initial sequents.
The system also enjoys a `local' cut-elimination procedure, via standard techniques, but that is beyond the scope of this work.
\end{remark}

\begin{proposition}\label{prop:LDTidentity}
    The following have polynomial size, cut-free,
    $\TreeLDT$ proofs:
    \begin{alphenumerate}[nolistsep]
        \begin{multicols}{2}
            \item \label{item:LDTidentity-A-A}
                $A \sequent A$
            \item $\sequent A, \overline A$
            \item $A, \overline A \sequent$
            \item $A \sequent p, ApB$
            \item $p, B \sequent ApB$
            \item $ApB \sequent A,p$
            \item $ApB, p \sequent B$
        \end{multicols}
    \end{alphenumerate}
\end{proposition}

\begin{proof}
To prove (a), we show by induction
on the complexity of~$A$ that $\Gamma, A \sequent A, \Delta$ has a
polynomial size, cut-free proof. In the base case, $A$~is
a literal~$p$, and this is an axiom. For the induction step,
$A$~has the form $BpC$, we use
\begin{center}
	{\small
    \Axiom$B, \Gamma \fCenter \Delta, B, p, p$
    \Axiom$p, C, \Gamma \fCenter \Delta, B, p$
    \BinaryInf$\Gamma, BpC \fCenter \Delta, B, p$
    \Axiom$B, p, \Gamma \fCenter \Delta, C, p$
    \Axiom$p, C, p, \Gamma \fCenter \Delta, C$
    \BinaryInf$p, \Gamma, BpC \fCenter \Delta, C$
    \BinaryInf$\Gamma, BpC \fCenter BpC, \Delta$
    \DisplayProof}
\end{center}
The first and fourth upper sequents are handled by
the induction hypothesis applied to $B$ and~$C$.
The second and third upper sequents obtained from axioms by weakenings.
By inspection, the resulting $\TreeLDT$ proof
has $O(n)$ lines each with $O(n)$ many symbols,
where $n$ is the size of~$A$.

Parts (b) and~(c) are proved similarly.
Parts (d)-(g) are now easy to prove with a single \emph{dec-l} or
\emph{dec-r} inference and invoking part~(a).
\end{proof}

\section{Comparing $\DT$ proof systems and $\LK$ proof systems}\label{sec:LK_and_LDT}

$\LK$ is the usual Gentzen sequent calculus for Boolean formulas
over the basis $\land$ and $\lor$. The \emph{Boolean formulas}
are defined inductively by
\begin{integerenumerate}
    \item Any literal $p$ is a Boolean formula, and
    \item If $A$ and $B$ are Boolean formulas, then
        so are $(A\lor B)$ and $(A \land B)$.
\end{integerenumerate}
The proof system $\LK$ has the same initial sequents (axioms)
as $\LDT$, its inference rules are
the contraction rules \emph{c-l} and \emph{c-r},
the weakening rules \emph{w-l} and \emph{w-r}, the cut rule, and
the following Boolean rules:
\begin{center}
\def\myExtraSpace{\hspace*{3em}}
\begin{tabular}{c@{\myExtraSpace}c}
\multicolumn{1}{l}{\bf Boolean rules:} \\[1ex]
\Axiom$A,B,\Gamma \fCenter \Delta$
\LeftLabel{\emph{$\landLeft$:}}
\UnaryInf$A\land B, \Gamma \fCenter \Delta$
\DisplayProof &
\Axiom$\Gamma \fCenter \Delta, A$
\Axiom$\Gamma \fCenter \Delta, B$
\LeftLabel{\emph{$\landRight$:}}
\BinaryInf$\Gamma \fCenter \Delta, A\land B$
\DisplayProof \\[4ex]
\Axiom$A,\Gamma \fCenter \Delta$
\Axiom$B,\Gamma \fCenter \Delta$
\LeftLabel{\emph{$\lorLeft$:}}
\BinaryInf$A\lor B, \Gamma \fCenter \Delta$
\DisplayProof &
\Axiom$\Gamma \fCenter \Delta, A, B$
\LeftLabel{\emph{$\lorRight$:}}
\UnaryInf$\Gamma \fCenter \Delta, A\lor B$
\DisplayProof \label{def:lorRules}
\end{tabular}
\end{center}

\begin{definition}
A \emph{clause} is a disjunction of literals; a \emph{term}
is a conjunction of literals. If $\vec p$ is
a vector of literals, we write
$\bigvee \vec p$ to denote any disjunction of the literals~$\vec p$,
taken in the indicated order. In other words, $\bigvee p_1$
denotes $p_1$; and $\bigvee \vec p$ denotes
any formula of the form $(\bigvee \vecpprime)\lor(\bigvee \vecppprime)$
where $\vecpprime$ and $\vecppprime$ denote
$p_1,\dots, p_k$ and $p_{k-1},\dots, p_\ell$ for some $1\le k \le \ell$.
The notation $\bigwedge \vec p$ is defined similarly.
\end{definition}

\begin{definition}
A Boolean formula is \emph{depth one} if it is either a clause or a term.
$\oneLK$ is the fragment of $\LK$ in which all formulas
appearing in sequents are depth one formulas. $\TreeoneLK$ is the
same system with the restriction that proofs are tree-like.
\end{definition}
Although the notations $\bigvee \vec p$ and $\bigwedge\vec p$ are ambiguous
about the nesting of disjunctions or conjunctions, this makes no
difference in our applications since, if
$A$ and $B$ are both of the form $\bigvee \vec p$ but
with different orders of applications of $\lor$'s,
then there are polynomial size, cut-free $\TreeoneLK$
proofs
of $A \sequent B$ and $B \sequent A$.

Later theorems
will compare the proof theoretic strengths
of various fragments and extensions of $\LDT$ to fragments of
$\LK$. Since these theories use different languages,
we need to establish translations between cedents of $\DT$ formulas
and (depth one) Boolean formulas.

\begin{definition}\label{prop:ConjDisj}
For a (nonempty) sequence of literals $\vec p$ we define the DT formulas $\Conj(\vec p) $ and $\Disj (\vec p) $ by induction on the length of $\vec p$ as follows:
\[
    \begin{array}{rcl}
         \Conj(p) & := & p \\
         \Conj(p, \vec p) & := & (pp\Conj(\vec p))
    \end{array}
    \qquad
    \begin{array}{rcl}
        \Disj(p) & := & p \\
        \Disj(p, \vec p) & := & (\Disj(\vec p) p p)
    \end{array}
\]
\end{definition}
\noindent
In other words, if $\vec p = (p_1$, \dots, $p_\ell)$, for $\ell > 1$, we have:
    \begin{eqnarray*}
        \Conj(\vec p) &=& (p_1 p_1 ( p_2 p_2 ( \cdots ( p_{\ell-2} p_{\ell - 2} ( p_{\ell-1} p_{\ell-1} p_\ell) ) \cdots ) ) ) \\[0.3ex]
        \Disj(\vec p) &=& ( ( ( \cdots ( ( p_\ell p_{\ell-1} p_{\ell-1}) p_{\ell-2} p_{\ell-2} ) \cdots ) p_2 p_2 ) p_1 p_1 ).
    \end{eqnarray*}
It is not hard to verify that $\Conj$ and $\Disj$ correctly
express the conjunction and disjunction of the literals $\vec p$.
This is borne out by the next proposition.

\begin{proposition}\label{prop:DTconjdisj}
The following sequents have polynomial size, cut-free $\TreeLDT$ proofs.
\begin{alphenumerate}[nolistsep]
\begin{multicols}{2}
    \item $\Conj(\vec p, \vec q) \sequent \Conj(\vec p)$
    \item $\Conj(\vec p, \vec q) \sequent \Conj(\vec q)$
    \item $\Conj(\vec p), \Conj(\vec q) \sequent \Conj(\vec p, \vec q)$
    \item $\Disj(\vec p) \sequent \Disj(\vec p, \vec q)$
    \item $\Disj(\vec q) \sequent \Disj(\vec p, \vec q)$
    \item $\Disj(\vec p, \vec q) \sequent \Disj(\vec p), \Disj(\vec q)$
\end{multicols}
\end{alphenumerate}
\end{proposition}
\begin{proof}
All six parts of the proposition are readily proved by
induction on the length of $\vec p$, applying
a \emph{dec-l} and \emph{dec-r} inference, and appealing
to the induction hypothesis. The base cases are handled with
the aid of Proposition~\ref{prop:LDTidentity}(a).
\end{proof}

For the converse direction of simulating $\LDT$ (and its supersystems) by
$\LK$, we need to express a DT formula~$A$ as Boolean formulas in both CNF
and DNF forms. For this we define $\Tms(A)$ as a multiset of terms
(i.e., a multiset of conjunctions)
and $\Cls(A)$ as a multiset of clauses (i.e., a multiset of disjunctions) so
that $A$ is equivalent to both the DNF $\bigvee \Tms(A)$
and the CNF $\bigwedge \Cls(A)$.
\begin{definition}\label{def:TmsCls}
Let $A$ be a $\DT$-formula. The \emph{terms} and \emph{clauses} of~$A$
are the
multisets $\Tms(A)$ and $\Cls(A)$ inductively defined by letting
$\Tms(p)$ and $\Cls(p)$ both equal~$p$, and letting
\begin{eqnarray}
\Tms(BpC) & := &
      \{ \overline p \land D : D \in \Tms(B) \} \cup
      \{ p \land D : D \in \Tms(C) \} \\[1ex]
\Cls(BpC) & := &
      \{ p \lor D : D \in \Cls(B) \} \cup
      \{ \overline p \lor D : D \in \Cls(C) \} .
\end{eqnarray}
The conjunctions and disjunctions are associated from right to left.
\end{definition}

It is clear from the definition that
the DNF $\bigvee \Tms(A)$ and the CNF $\bigwedge \Cls(A)$ are both
equivalent to~$A$.
\begin{proposition}
\label{prop:sigma-pi-props}
    For $\DT$ formulas $A$ and $B$,
    there are polynomial size, cut-free $\TreeLK$-proofs of:
    \begin{alphenumerate}
        \item
            $C\sequent D$, for each $C \in \Tms(A)$ and $D\in \Cls(A)$.
        \item
            \begin{enumerate}[label={\rm (\roman*)}, ref=(\roman*)]
                \item
                    $\Cls (ApB) \sequent D,p $, for each $D \in \Cls (A)$;
                \item
                    $p,\Cls (ApB) \sequent D$, for each $D \in \Cls(B)$.
                \item $\Cls(A) \sequent D,p$, for each $D \in \Cls (ApB)$.
                \item $p,\Cls(B) \sequent D$, for each $D \in \Cls(ApB)$.
            \end{enumerate}
        \item
            \begin{enumerate}[label={\rm (\roman*)}, ref=(\roman*)]
                \item $C \sequent p,\Tms (ApB)$,
    			    for each $C \in \Tms(A)$;
                \item $p,C \sequent \Tms(ApB)$,
                    for each $C\in \Tms(B)$.
                \item $C \sequent p, \Tms(A)$, for each $C \in \Tms(ApB)$.
                \item $p,C \sequent \Tms (B)$, for each $C \in \Tms(ApB)$.
            \end{enumerate}
    \end{alphenumerate}
\end{proposition}
Part~(a) of the lemma is proved by induction on the
complexity of $A$.  Parts (b) and~(c)
are trivial once the definitions are unwound. For example,
(b.i)~follows from the fact that $\Cls(ApB)$ contains the formula $p\lor D$. This allows (b.i) to be derived from the two sequents
$p\sequent p$ and $D\sequent D$. The former is an axiom, and
the latter has a tree-like cut-free proof by
Proposition~\ref{prop:LDTidentity}\ref{item:LDTidentity-A-A}.
The other cases are similar.

\begin{proposition}\label{prop:ClsImpliesTms}
There are polynomial size atomic-cut $\TreeLK$ proofs
and polynomial size cut-free $\LK$ proof of the sequents
$\Cls(A)\sequent \Tms(A)$ for $\DT$ formulas~$A$.
\end{proposition}

\begin{proof}
We prove the tree-like case by giving a recursive construction.
Assume $A$ is $BpC$. We claim that there is a polynomial
size tree-like $\LK$ derivation~$\pi_0$ of the sequent
\begin{equation}\label{eq:ClsImpliesTms_B}
\{ p \lor D : D \in \Cls(B) \} \sequent \{ \overline p \lor D : D \in \Tms(B) \}, p
\end{equation}
which uses a single instance $\Cls(B)\sequent \Tms(B)$ as a non-logical
initial sequent. Indeed, $\pi_0$ is easily constructed by combining
$\Cls(B)\sequent \Tms(B)$ with initial sequents $p \sequent p$
and $\sequent p, \overline p$ using $\lorLeft$ and $\lorRight$
inferences. Similarly, there is a polynomial size $\TreeLK$
proof of
\begin{equation}\label{eq:ClsImpliesTms_C}
p, \{ \overline p \lor D : D \in \Cls(C) \} \sequent \{ p \lor D : D \in \Tms(C) \}
\end{equation}
which uses a single instance $\Cls(C)\sequent \Tms(C)$ as a non-logical
initial sequent. Combining (\ref{eq:ClsImpliesTms_B}) and (\ref{eq:ClsImpliesTms_C})
with a cut on~$p$ gives a tree-like $\LK$ derivation of
$\Cls(A) \sequent \Tms(A)$ which uses single instances of the sequents
$\Cls(B) \sequent \Tms(B)$ and $\Cls(C)\sequent \Tms(C)$ as non-logical
initial sequents. Proceeding recursively gives the desired
polynomial size atomic-cut $\TreeLK$ proof of $\Cls(A)\sequent \Tms(A)$.

It is straightforward to give (dag-like) cut-free $\LK$
polynomial size proof of $\Cls(A)\sequent \Tms(A)$, and this is
omitted. Alternatively, \cite{Buss:cutElimInSitu} gives a
general construction that, given a tree-like $\LK$ proof
in which all cuts are atomic, forms a linear size dag-like
$\LK$ proof.
\end{proof}

Proposition~\ref{prop:ClsImpliesTms} can be extended to show that
there are quasipolynomial size cut-free $\TreeLK$ proofs of
$\Cls(A) \sequent \Tms(A)$, but it is open whether polynomial size
is possible.

The next definition shows how to compare proof complexity between proof
systems that work with DT formulas and ones that work with Boolean formulas.

\begin{definition}\label{def:simulations}
    Let $P$ be a proof system for sequents of Boolean formulas (or at least,
    sequents of depth one Boolean formulas),
    and $Q$ be a proof system for sequents of DT formulas.
    We say that $P$ \emph{polynomially simulates}~$Q$ if
    there is a polynomial time procedure which, given a $Q$-proof
    of
    \begin{equation}\label{eq:fromDT}
        A_0, \dots , A_{m-1} \sequent B_0 , \dots , B_{n-1} ,
    \end{equation}
    where the $A_i$'s and $B_i$'s are $\DT$-formulas,
    produces a $P$-proof of
    \begin{equation}\label{eq:toLK}
        \Cls(A_0), \dots, \Cls (A_{m-1}) \sequent \Tms(B_0), \dots , \Tms(B_{n-1}) .
    \end{equation}
    The system $Q$ \emph{polynomially simulates} $P$ if there is
    a polynomial time procedure which, given a $P$-proof of
    \begin{equation}\label{eq:fromLK}
        \bigvee \vec a_0, \dots , \bigvee \vec a_{m-1} \sequent \bigwedge \vec b_0, \dots , \bigwedge \vec b_{n-1},
    \end{equation}
    where the $\vec a_i$'s and $\vec b_i$'s are sequences of literals,
    produces a $Q$-proof of
    \begin{equation}\label{eq:toDT}
        \Disj(\vec a_0), \dots , \Disj(\vec a_{m-1})\sequent \Conj(\vec b_0), \dots , \Conj(\vec b_{n-1}) .
    \end{equation}
    The systems $P$ and $Q$ are \emph{polynomially equivalent}
    if they polynomially simulate each other.
    (\ref{eq:toLK}) is called the \emph{Boolean translation} of (\ref{eq:fromDT}).
    (\ref{eq:toDT}) is called the \emph{$\DT$-translation} of (\ref{eq:fromLK}).
    \emph{Quasipolynomial simulation and equivalence} are defined in the same way, but
    using quasipolynomial time (time $2^{\log^{O(1)}n}$) procedures.\footnote{It turns out that all stated quasipolynomial simulations in this work (Theorems~\ref{thm:1lk-sim-ldt} and \ref{thm:LK_ELNDT}) take time $n^{O(\log n)} = 2^{O(\log^2 n)}$.}
\end{definition}

\subsection{$\oneLK$ and $\LDT$}

\begin{theorem}\label{thm:ldt-sim-1lk}
$\LDT$ polynomially simulates $\oneLK$. $\TreeLDT$ polynomially simulates $\TreeoneLK$.
\end{theorem}

\begin{proof}
Suppose $\pi$ is a $\oneLK$ proof.  Every formula
in~$\pi$ is either a term $\bigwedge \vec a$ or a clause
$\bigvee \vec a$, where $\vec a$ is a vector of literals.
We modify $\pi$ by replacing each such formula
by $\Conj(\vec a)$ or $\Disj(\vec a)$ respectively.
The initial sequents and the contraction, weakening and cut
inferences in~$\pi$ become valid initial sequents or
contraction, weakening and cut inferences for $\LDT$.

An $\landLeft$ inference in~$\pi$ of the form
\begin{prooftree}
\Axiom$\bigwedge \vec a, \bigwedge \vec b ,\Pi\fCenter \Delta$
\LeftLabel{\emph{$\landLeft$:}}
\UnaryInf$\bigwedge \vec a \land \bigwedge \vec b ,\Pi\fCenter \Delta$
\end{prooftree}
is replaced by
\begin{equation}\label{eq:andLeftSim}
\hbox{
\Axiom$\Conj( \vec a ), \Conj(\vec b), \Pi \fCenter \Delta$
\UnaryInf$\Conj(\vec a, \vec b), \Pi \fCenter \Delta$
\DisplayProof }
\end{equation}
This is not a valid $\LDT$ inference. To fix this, note that by
parts (a) and~(c) of Proposition~\ref{prop:DTconjdisj},
the cedents $\Conj(\vec a, \vec b) \sequent \Conj(\vec a)$
and $\Conj(\vec a, \vec b) \sequent \Conj(\vec a)$
have polynomial-size (cut-free) $\TreeLDT$ proofs.
Using two cut inferences with these sequents
gives a valid $\LDT$ derivation of the lower
sequent of~(\ref{eq:andLeftSim}) from the upper sequent.

An $\landRight$ inference in~$\pi$ of the form
\begin{prooftree}
\Axiom$\Pi \fCenter \Delta, \bigwedge \vec a$
\Axiom$\Pi \fCenter \Delta, \bigwedge \vec b$
\LeftLabel{\emph{$\landRight$:}}
\BinaryInf$\Pi \fCenter \Delta, \bigwedge \vec a \land \bigwedge \vec b$
\end{prooftree}
is replaced by
\begin{equation}\label{eq:andRightSim}
\hbox{
\Axiom$\Pi \fCenter \Delta, \Conj(\vec a)$
\Axiom$\Pi \fCenter \Delta, \Conj(\vec b)$
\BinaryInf$\Pi \fCenter \Delta, \Conj(\vec a, \vec b)$
\DisplayProof }
\end{equation}
The sequent $\Conj(\vec a), \Conj(\vec b) \sequent \Conj(\vec a, \vec b )$
has a polynomial size (cut-free) $\TreeLDT$ proof by
Proposition~\ref{prop:DTconjdisj}(e). Cutting the two upper
sequents of~(\ref{eq:andRightSim}) against this gives a valid
$\TreeLDT$ derivation of the
lower sequent.

Dual constructions allow $\lorLeft$ and $\lorRight$ inferences in~$\pi$
to be converted into valid $\TreeLDT$ derivations. The result is a valid
$\LDT$ proof~$\pi^\prime$ of the DT-translation of the final line
of~$\pi$. By construction, $\pi^\prime$ has size polynomially bounded by the
size of~$\pi$.
Since the upper sequents of (\ref{eq:andLeftSim}) and~(\ref{eq:andRightSim})
were used only once when forming the $\TreeLDT$ derivations simulating
inferences of~$\pi$, the $\LDT$ proof~$\pi^\prime$ is tree-like
whenever $\pi$ is tree-like.
\end{proof}

A converse result holds too, but we have only a quasipolynomial
simulation in the tree-like case. It is open whether this can be improved
to a polynomial simulation.

\begin{theorem}\label{thm:1lk-sim-ldt}
$\oneLK$ polynomially simulates $\LDT$.
$\TreeoneLK$ quasipolynomially
simulates $\TreeLDT$.
\end{theorem}

\begin{proof}
Suppose $\pi$ is an $\LDT$ proof, possibly tree-like.
We need to convert $\pi$ into a $\oneLK$ proof~$\pi^\prime$.
As a first step, each sequent in $\pi$ is replaced
by its Boolean translation as defined in~(\ref{eq:toLK}).
Namely, every $\DT$ formula~$A$ in the antecedent,
of a sequent in~$\pi$ is replaced
by the cedent $\Cls(A)$;
and every $\DT$ formula~$A$ in a succedent is replaced by the
cedent $\Tms(A)$.
Since $\Cls(p)$ and $\Tms(p)$ are both equal to~$p$,
the Boolean translation of an axiom in $\pi$ is
a valid $\LK$ axiom. Likewise, any contraction
or weakening inference in $\pi$ is readily replaced valid
$\LK$ inferences after forming the Boolean translations.
The decision rules and cut rules in~$\pi$, however,
need to be fixed up to make $\pi^\prime$ a valid
$\LK$-proof.

First consider a \emph{dec-l} inference in~$\pi$
\begin{equation}\label{eq:LDTbyoneLKorig}
\hbox{\Axiom$A, \Gamma \fCenter \Delta, p$
\Axiom$p, B, \Gamma \fCenter \Delta$
\LeftLabel{\emph{dec-l:}}
\BinaryInf$ApB, \Gamma \fCenter \Delta$
\DisplayProof}
\end{equation}
The Boolean translation of this gives
\begin{equation}\label{eq:LDTbyoneLK}
\hbox{\Axiom$\Cls(A), \Gamma^* \fCenter \Delta^*, p$
\Axiom$p, \Cls(B), \Gamma^* \fCenter \Delta^*$
\BinaryInf$ApB, \Gamma^* \fCenter \Delta^*$
\DisplayProof}
\end{equation}
where $\Gamma^*\sequent\Delta^*$ is the Boolean translation
of $\Gamma \sequent \Delta$.
Let $\Cls(A)$ equal $D_1,\dots, D_\ell$,
and $\Cls(B)$ equal $E_1,\dots, E_k$, so that
that $\Cls(ApB)$ equals the union of
$\{ p\lor\penalty10000 D_i\}_{i\le \ell}$
and $\{ \overline p \lor\penalty10000 E_i \}_{i \le k}$.
Starting with the upper left sequent of~(\ref{eq:LDTbyoneLK}),
we form an $\ell$ step tree-like derivation
\begin{equation}\label{eq:pClsA}
\hbox{\Axiom$p \fCenter p$
\Axiom$\Cls(A), \Gamma^* \fCenter \Delta^*, p$
\LeftLabel{$\ell$ many \emph{$\lorLeft$'s:}}
\doubleLine
\BinaryInf$\{ p\lor D_i\}_{i\le \ell}, \Gamma^* \fCenter \Delta^*, p$
\DisplayProof}
\end{equation}
This derivation uses $\ell$ instances of the axiom $p\sequent p$
and $\ell$ inferences of the form
\begin{prooftree}
\Axiom$p \fCenter p$
\Axiom$\{ p\lor D_i\}_{i<j}, D_j, \{D_i\}_{i > j}, \Gamma^* \fCenter \Delta^*, p$
\LeftLabel{\emph{$\lorLeft$:}}
\BinaryInf$\{ p\lor D_i\}_{i<j}, p \lor D_j, \{D_i\}_{i > j}, \Gamma^* \fCenter \Delta^*, p$
\end{prooftree}
A similar $k$ step tree-like $\LK$ proof derives
\begin{equation}\label{eq:pClsB}
\hbox{\Axiom$p, \overline p \fCenter $
\Axiom$p, \Cls(B), \Gamma^* \fCenter \Delta^*$
\LeftLabel{$k$ many \emph{$\lorLeft$'s:}}
\doubleLine
\BinaryInf$p, \{\overline p\lor E_i\}_{i\le k}, \Gamma^* \fCenter \Delta^*$
\DisplayProof}
\end{equation}

\noindent
Combining (\ref{eq:pClsA}) and~(\ref{eq:pClsB}) with a
cut on the atomic formula~$p$ gives the lower
sequent, $\Cls(ApB) \Gamma \sequent \Delta$,
of (\ref{eq:LDTbyoneLK}) as desired.
This gives
a tree-like $\LK$-derivation simulating
(\ref{eq:LDTbyoneLK}) of size polynomially bounded
by the size of the lower sequent of~(\ref{eq:LDTbyoneLKorig}).

The case of a \emph{dec-r} inference in~$\pi$ is
handled dually; we omit the argument.

Now consider a cut inference in~$\pi$:
\begin{equation}\label{eq:LDTbyoneLKcutBefore}
\hbox{%
\Axiom$ \Gamma \fCenter \Delta, A$
\Axiom$ A, \Gamma \fCenter \Delta$
\BinaryInf$ \Gamma \fCenter \Delta$
\DisplayProof}
\end{equation}
The Boolean translation of this is
\begin{equation}\label{eq:LDTbyoneLKcut}
\hbox{\Axiom$ \Gamma^* \fCenter \Delta^*, \Tms(A)$
\Axiom$ \Cls(A), \Gamma^* \fCenter \Delta^*$
\BinaryInf$ \Gamma^* \fCenter \Delta^*$
\DisplayProof}
\end{equation}
Again let $\Cls(A)$ be $\{D_i\}_{i\le \ell}$;
and let $\Tms(A)$ be $\{F_i\}_{i\le m}$.
By Lemma~\ref{prop:sigma-pi-props}(a),
there are short cut-free $\TreeLK$ proofs
for each $F_i\sequent D_j$. The strategy for
converting (\ref{eq:LDTbyoneLKcut}) a
valid $\LK$-derivation is to repeatedly
cut with these sequents.

There are two ways to do this.
The first construction starts by deriving, for each~$i$,
the clause $F_i, \Gamma^* \sequent \Delta^*$
by using $\ell$ cut inferences
combining the sequents $F_i\sequent D_j$ (for $j\le \ell$)
against the upper right sequent of~(\ref{eq:LDTbyoneLKcut}).
Then, combining these sequents with $m$ cuts against
the upper left sequent of~(\ref{eq:LDTbyoneLKcut})
gives the desired sequent $\Gamma^* \sequent \Delta^*$.

The second, alternative, construction is dual.
It starts by deriving, for each~$j$,
the clause $\Gamma^* \sequent\penalty10000 \Delta^*,D_j$
by using $m$ cuts inferences
combining the sequents $F_i\sequent\penalty10000 D_j$ (for $i\le m$)
against the upper left sequent of~(\ref{eq:LDTbyoneLKcut}).
Then, combining these sequents with $\ell$ cuts against
the upper right sequent of~(\ref{eq:LDTbyoneLKcut})
gives the desired sequent $\Gamma^* \sequent \Delta^*$.

Either of these constructions gives immediately a polynomial-size
$\oneLK$ derivation simulating the
inference (\ref{eq:LDTbyoneLKcut}). The first construction
is not tree-like since it
uses the upper right sequent of (\ref{eq:LDTbyoneLKcut})
$m$ times. Likewise, the second construction used the upper left sequent
$\ell$ times. But in either case, this yields a dag-like
derivation, completing the polynomial
simulation of $\LDT$ by $\oneLK$.

The same constructions can work for the tree-like case, but this
requires a more careful size analysis and
gives only a quasipolynomial simulation.
If $\pi$ ends with a \emph{dec-l} and \emph{dec-r}
inference, let $\pi_0$ and~$\pi_1$
be the subderivations of~$\pi$ that end with the
upper left and right sequents (respectively) of
the inference~(\ref{eq:LDTbyoneLKorig}).
We use $\pi^*$, $\pi_0^*$ and $\pi_1^*$ to denote the
$\TreeoneLK$ proofs obtainable by the constructions above.
As $\pi$ ends with a decision inference, inspection of the construction
above shows
\[
|\pi^*| ~\le~ |\pi_0^*| + |\pi_1^*| + n^{O(1)}.
\]

Now suppose that $\pi$ ends with the cut inference~(\ref{eq:LDTbyoneLKcutBefore}),
and let $\pi_0$ and~$\pi_1$
be the subderivations of~$\pi$ that end with the
upper left and right sequents of~(\ref{eq:LDTbyoneLKcutBefore}).
If $|\pi_1|\le |\pi_0|$, then $|\pi_1| < |\pi|/2$; in this case, use the first construction that uses $\pi_0^*$ once and $\pi_1^*$ $m$~times,
to obtain a tree-like $\pi^*$ of size bounded by
$|\pi_0^*| + O(m \cdot |\pi_1^*|)$.  Dually, if
$|\pi_0|\le |\pi_1|$, then $|\pi_0|<|\pi|/2$
and the second construction yields $\pi^*$ of size bounded by
$|\pi_1^*| + O(\ell \cdot |\pi_0^*|)$.

Let $S(n)$ be the minimal size $\TreeoneLK$ proof required to simulate a $\TreeLDT$ proof~$\pi$ of size~$n$, namely $|\pi^*| \le S(|\pi|)$.
Combining the above size bounds into a single (rather crude)
estimate and letting $S(0)=0$ gives, for each $n$,
values $a$ and $b$ such that $a+b<n$ and
\[
S(n) ~\le~ S(a)+S(b)+ n^{O(1)} S(n/2).
\]
From this $S(n) = n^{O(\log n)}$ follows immediately,
giving the desired quasipolynomial simulation.
\end{proof}

\section{Nondeterministic decision trees and \texorpdfstring{$\LNDT$}{LNDT}}\label{sec:NDTsystems}

This section defines nondeterministic decision tree ($\NDT$) formulas,
and the associated sequent calculus $\LNDT$. The $\NDT$ formulas have
two kinds of connectives; the 3-ary case function $ApB$ and the Boolean
or gate ($\lor$). Formally,
\begin{definition}\label{def:NDT}
The \emph{nondeterministic decision tree formulas}, or \emph{$\NDT$ formulas}
for short, are inductively defined by
\begin{integerenumerate}
\item Any literal $p$ is a $\NDT$ formula, and
\item If $A$ and $B$ are $\NDT$ formulas
and $p$ is a variable, then $(A p B)$ is a $\NDT$ formula.
\item If $A$ and $B$ are $\NDT$ formulas, then
$(A\lor B)$ is an $\NDT$ formula.
\end{integerenumerate}
\end{definition}

A nondeterministic gate in a decision tree means a gate
which is accepting exactly when at least one of its children
is accepting. The corresponds exactly to an $\lor$ gate, which
yields \emph{True} exactly when at least one input is \emph{True}.
One of our motivations in defining
$\LNDT$ that is will serve as a foundation for our later
definition $\eLNDT$, which will capture a logic
for nondeterministic branching programs, and hence a logic
for nonuniform NL.

\begin{definition}\label{def:LNDT}
The sequent calculus $\LNDT$ is a proof system
in which lines are sequents of $\NDT$ formulas.
The valid initial sequents (axioms) and rules are the same
as those of $\LDT$ (Definition~\ref{def:DT}),
along with the two $\lor$ inferences, $\lorLeft$ and $\lorRight$
of $\LK$ as described
on page~\pageref{def:lorRules}.
\end{definition}

For $\alpha$ a 0-1-truth assignment, the
semantics of $\NDT$ formulas is defined extending the
definition of the semantics of $\DT$ formulas, in equations~\ref{eq:TruthDefnDT},
to include
\[
\alpha( A\lor B) = \begin{cases}
    1 \quad & \hbox{if $\alpha(A) = 1$ or $\alpha(B) = 1$} \\
    0       & \text{otherwise.}
\end{cases}
\]
It is straightforward to verify that $\LNDT$ is
implicationally sound and implicationally complete
for sequents of $\NDT$ formulas.

An important fact for $\NDT$ formulas is that we can, without
loss of much generality, require the $\lor$'s to be used only
as topmost connectives. This is formalized by the following definitions
and theorem.
\begin{definition}
\label{def:NDTnormalform}
An $\NDT$~$A$ is in \emph{normal form} if it has
the form $\bigvee_{i<n} A_i$
where each $A_i$ is a $\DT$ formula, i.e.,
each $A_i$ is $\vee$-free.
\end{definition}
As we show below, the fact that $\NDT$ are formulas (not circuits) means that there is a polynomial time procedure to transform a
a $\NDT$ formula to normal form.
\begin{definition}\label{def:NDTnormalFormFormula}
We extend the definition of the multiset $\Tms(A)$ to $\NDT$ formulas~$A$,
by inductively defining
\begin{eqnarray*}
\Tms(B p C) & := &
     \{ \overline p \land D : D \in \Tms(B) \} \cup
  \{ p \land D : D \in \Tms(C) \} \\[1ex]
\Tms(B \lor C ) & := & \Tms(B) \cup \Tms(C).
\end{eqnarray*}
The multiset $\DTms(A)$
is defined to be the set
of $\DT$ formulas
\[
\DTms(A) ~=~ \{ \Conj(\vec p) : {\textstyle\bigwedge \vec p \in \Tms(A)} \}.
\]
Equivalently, $\DTms(B\lor C) = \DTms(B) \cup \DTms(C)$ and
\[
\DTms(BpC) ~=~ \{ (\overline p\, \overline p\, D) : D \in \DTms(B) \} \cup
               \{ (p p D) : D \in \DTms(C) \}.
\]
The normal form of an $\NDT$ formula~$A$ is defined
to equal $\NF(A) := \bigvee \DTms(A)$.  The
disjunction consists of binary $\lor$ gates
applied the members of $\DTms(A)$. For convenience,
the disjunctions are ordered to respect the
structure of the formula~$A$. In particular,  $\NF(A \lor B)$
is just $\NF(A)\lor \NF(B)$.
\end{definition}

The next proposition formalizes the intuition that $\NF(A)$
is equivalent to~$A$.

\begin{proposition}\label{prop:ndtNF}
    The following have polynomial size, cut-free
    $\TreeLNDT$ proofs:
    \begin{multicols}{2}
        \begin{alphenumerate}[nolistsep]
            \item $\NF(A)\sequent p, \NF(ApB)$
            \item $p , \NF(B)\sequent \NF(ApB)$
            \item $\NF(ApB) \sequent \NF(A), p$
            \item $p, \NF(ApB) \sequent \NF(B)$
            \item $\NF(A) \sequent \NF(A\lor B)$
            \item $\NF(B) \sequent \NF(A\lor B)$
            \item $\NF(A \lor B) \sequent \NF(A), \NF(B)$
        \end{alphenumerate}
    \end{multicols}
\end{proposition}

\begin{proof}
We first prove (a); parts (b)-(d) are similar.
For each formula $D$ in $\DTms(A)$, the
sequent $D\sequent D$ has a polynomial size cut-free
$\TreeLDT$ proof by Proposition~\ref{prop:LDTidentity}(a).
From this, derive in $\LDT$,
\begin{center}
\Axiom$\fCenter p, \overline p$
\doubleLine
\LeftLabel{\emph{w-l, w-r:}}
\UnaryInf$D\fCenter \overline p,\overline p, p$
\Axiom$D \fCenter D$
\doubleLine
\LeftLabel{\emph{w-l, w-r:}}
\UnaryInf$\overline p, D \fCenter p, D$
\LeftLabel{\emph{dec-r:}}
\BinaryInf$D \fCenter p, (\overline p \, \overline p\, D)$
\DisplayProof
\end{center}
Combining all the sequents $D \sequent p, (\overline p \, \overline p\, D)$
with a tree of $\lorLeft$, $\lorRight$
and weakening inferences gives the desired
sequent $\NF(A)\sequent p, \NF(ApB)$.

To prove (e)-(g), note again that for each $D\in\Tms(A\lor B)$,
there is a polynomial size, cut-free proof of $D\sequent D$.
Then each of (e)-(g) can be derived by
combining (some of) these sequents with a tree of $\lor$
and weakening inferences.
\end{proof}

We write $\LNDT^\NF$ to denote the proof
system $\LNDT$ restricted to use sequents
containing only $\NDT$ formulas in normal form.

\begin{theorem}\label{thm:LNDT_NF_wlog}
Suppose $\Gamma\sequent\Delta$ contains only normal form
$\NDT$ formulas. Suppose $\pi$ is an
$\LNDT$ (respectively, a $\TreeLNDT$) proof of $\Gamma\sequent\Delta$.
Then $\Gamma \sequent \Delta$ has an
$\LNDT^\NF$ (respectively, a $\TreeLNDT^\NF$) proof~$\pi^\prime$
of size polynomially bounded by the size of~$\pi$.
\end{theorem}

\begin{proof}
As a first step towards forming $\pi^\prime$, replace
every formula~$A$ in~$\pi$ with $\NF(A)$.  Axioms in~$\pi$
are unchanged. Contraction inferences, weakening inferences,
and cut inferences in~$\pi$ remain valid inferences. Likewise,
since $\NF(A\lor B)$ equals $\NF(A) \lor \NF(B)$, the
$\lor$ inferences in~$\pi$ remain valid.  However, the
\emph{dec-r} and \emph{dec-l} may no longer be valid and need
to be fixed up.  Consider a \emph{dec-r} inference in~$\pi$:
\[
\hbox{
\Axiom$\Pi \fCenter \Lambda, A, p$
\Axiom$p, \Pi \fCenter \Lambda, B$
\LeftLabel{\emph{dec-r:}}
\BinaryInf$\Pi \fCenter \Lambda, ApB$
\DisplayProof }
\]
This is transformed to
\begin{equation}\label{eq:LNDT_NF_wlog_Pf}
\hbox{
\Axiom$ \Pi^* \fCenter \Lambda^*, \NF(A), p$
\Axiom$p, \Pi^* \fCenter \Lambda^*, \NF(B)$
\doubleLine
\BinaryInf$\Pi^* \fCenter \Lambda^*, \NF(ApB)$
\DisplayProof }
\end{equation}
where $\Pi^*$ and $\Lambda^*$ are the cedents obtained
after replacing each formula by its normal form.
Applying cuts with the formulas (a) and~(b)
of Proposition~\ref{prop:ndtNF} and then a cut on~$p$
gives
\[
\small
\hbox{
\Axiom$ \Pi^* \fCenter \Lambda^*, \NF(A), p$
\Axiom$\NF(A)\fCenter p, \NF(ApB)$
\doubleLine
\BinaryInf$\Pi^* \fCenter \Lambda^*, \NF(ApB), p$
\Axiom$p, \Pi^* \fCenter \Lambda^*, \NF(B)$
\Axiom$p , \NF(B)\fCenter \NF(ApB)$
\doubleLine
\BinaryInf$p, \Pi^* \fCenter \Lambda^*, \NF(ApB)$
\BinaryInf$\Pi^* \fCenter \Lambda^*, \NF(ApB)$
\DisplayProof}
\]
This turns (\ref{eq:LNDT_NF_wlog_Pf}) into a $\LNDT$ derivation.
\end{proof}

\subsection{$\LDT$ and tree-like $\LNDT$ are equivalent}
\label{sec:LdtTreeLndt}

Next we turn to the relative complexity of $\LDT$ and $\LNDT$.
Naturally the latter subsumes the former,
but this can be strengthened as follows.\footnote{%
This also refines the known polynomial equivalence between
$\oneLK $ and $\TreeSys{\dLK{2}}$,
cf.\ Figure~\ref{fig:simulations}.}

\begin{theorem}\label{thm:tree-lndt-equiv-ldt}
$\TreeLNDT$ is polynomially equivalent to $\LDT$ over $\DT$-sequents.
\end{theorem}

\begin{proof}
We first show $\TreeLNDT$ polynomially simulates $\LDT$. Suppose $\pi$ is
an $\LDT$-proof (possibly dag-like) with $m$ sequents
$\Gamma_i\sequent\Delta_i$ for $i=1,\dots,m$. Define
$\overline \Gamma$ to be the multiset of formulas $\overline F$
for $F\in\Gamma$.  Let $A_i$ be $\bigvee(\overline\Gamma\cup\Delta)$,
namely a tree of (binary) disjunctions of the formulas
in $\overline\Gamma \cup\Delta$. (The disjunctions may be applied
in any order.) Clearly, each $A_i$ is
a $\NDT$-formula.

The next claim  will help us work with
disjunctions.
\begin{claim}\label{clm:disjunctions}
Let $\Pi, \Lambda,\Gamma, \Delta$ be cedents.
Suppose that for each formula $F\in \Pi$,
the formula $F\in \Lambda\cup \Delta\cup \overline\Gamma$.
(If there are multiple occurrences
of $F$ in~$\Pi$ it is not required to have multiple occurrences of~$F$
in $\Lambda\cup \Delta\cup \overline\Delta$.)
Let $\cal H$ (``hypotheses'') be the set containing the cedents
$F\sequent F$ such that $F\in \Pi\cap(\Lambda\cup \Delta)$
and the cedents $\overline F, F \sequent$ for $F\in(\Pi\cap\overline\Gamma)$.
Then the sequent
\[
\Gamma, \bigvee \Pi \sequent \bigvee\Lambda, \Delta
\]
has
a polynomial size, tree-like, cut-free proof from (a subset of) the initial sequents~$\cal H$.
using only $\lor$ inferences and weakenings.
\end{claim}

To understand the claim, note that the assumption is that any $F$ in~$\Pi$
also appears in $\Lambda$ or $\Delta$ or negated in~$\Gamma$.
The proof of the claim is by a simple application of
$\lorLeft$ and $\lorRight$ rules.

Returning to the proof of Theorem~\ref{thm:tree-lndt-equiv-ldt},
consider some $A_i$. If $\Gamma_i\sequent \Delta_i$ is an axiom,
then $A_i$ has
the form $p \lor \overline p$.  Clearly there is a short cut-free
$\TreeLNDT$ proof of $\sequent A_i$.  If
$\Gamma_i\sequent\Delta_i$ is inferred from
$\Gamma_j\sequent\Delta_j$ by a \emph{unary} inference (with $j<i$),
then by inspection of the contraction and weakening rules,
$(\overline \Gamma_j \cup \Delta_j)\subseteq\overline \Gamma_i \cup \Delta_i$.
Thus, by the claim,
there is a polynomial size, cut-free $\TreeLNDT$-proof of
$A_j\sequent A_i$,
since $A_i$ is $\bigvee(\overline\Gamma_i \cup \Delta_i)$
and $A_j$ is $\bigvee(\overline\Gamma_j \cup \Delta_j)$.

Finally, suppose $\Gamma_i\sequent\Delta_i$ is inferred
by a binary inference from
$\Gamma_j\sequent\Delta_j$ and $\Gamma_k\sequent\Delta_k$
(with $j,k<i$).
We will prove that the sequent$A_j, A_k\sequent A_i$
has a polynomial size tree.
Suppose $A_i$ is inferred by a cut inference,
\[
\hbox{
\Axiom$\Gamma_i \fCenter\Delta_i, C$
\Axiom$C, \Gamma_i \fCenter, \Delta_i$
\LeftLabel{\emph{cut:}}
\BinaryInf$\Gamma_i \fCenter \Delta_i$
\DisplayProof }
\]
Then $A_j$ is $\bigvee(\Delta_i\cup\{C\}\cup\overline\Gamma_i)$
and $A_i$ is $\bigvee(\Delta_i\cup\overline\Gamma_i)$
and the Claim~\ref{clm:disjunctions} and
Proposition~\ref{prop:LDTidentity} imply that
$A_j \sequent A_i, C$ has a polynomial size cut-free
proof.  Similarly, $\overline C, A_k \sequent A_i$ has
polynomial size, cut-free proof. Using a cut on~$C$,
gives a proof of $A_j, A_k\sequent A_i$.
Second, suppose $A_i$ is inferred by a
\emph{dec-l} inference
\[
\hbox{
\Axiom$A,\Gamma_i^\prime \fCenter \Delta_i, p$
\Axiom$p, B, \Gamma_i^\prime \fCenter \Delta_i$
\LeftLabel{\emph{dec-l:}}
\BinaryInf$ApB, \Gamma_i^\prime \fCenter \Delta_i$
\DisplayProof }
\]
where $\Gamma_i$ is $ApB,\Gamma_i^\prime$, and the upper left
and right sequents are
$\Gamma_j\sequent\Delta_j$ and $\Gamma_k\sequent\Delta_k$,
respectively.
Since $A_j$ is $\bigvee \{ \overline A, \overline \Gamma_i, \Delta_i, p \}$ and
$A_k$ is $\bigvee \{ \overline B, \overline p, \overline \Gamma_i, \Delta_i \}$
and $A_i$ is $\bigvee \{ \overline {ApB}, \overline \Gamma_i, \Delta_i \}$,
Claim~\ref{clm:disjunctions} and Proposition~\ref{prop:LDTidentity}
give polynomial size, cut-free $\TreeLNDT$ proofs of
$A, A_j \sequent A_i, p$ and
$p, B, A_k \sequent A_i$.  Applying a \emph{dec-l} rule
gives a polynomial size $\TreeLNDT$ of $A_j,A_k\sequent A_i$.
The third case where $A_i$ is inferred by a
\emph{dec-l} inference is similar, and again we
obtain a polynomial size $\TreeLNDT$ of $A_j,A_k\sequent A_i$.

We have shown that for each $i\le m$, there is are (up to two) values
$j,k<\penalty10000 i$
such that the sequent $A_j, A_k\sequent A_i$ has a polynomial
size, $\TreeLNDT$ proof, where the formulas $A_j$ and $A_k$ are
possibly omitted.  We can now complete the proof of the first half of
Theorem~\ref{thm:tree-lndt-equiv-ldt}. By Claim~\ref{clm:disjunctions},
there is a polynomial size $\TreeLNDT$ proof of
$A_1,\dots, A_m, \Gamma_m \sequent \Delta_m$. Cutting
with the sequents $A_j, A_k\sequent A_i$
for $i=m, m{-}1, \dots, 2, 1$, we derive
successively $A_1,\dots, A_\ell, \Gamma_m \sequent \Delta_m$
for $\ell = m, \dots, 2, 1$.  With $\ell=0$, a polynomial size
$\TreeLNDT$ proof of $\Gamma_m \sequent \Delta_m$, the endsequent
of~$\pi$.  This completes the proof that $\TreeLNDT$ polynomially simulates
$\LDT$.

To prove the second part of Theorem~\ref{thm:tree-lndt-equiv-ldt},
suppose $\pi$ is a $\TreeLNDT$ proof.
By Theorem~\ref{thm:LNDT_NF_wlog}, we may assume that every
formula in~$\pi$ is in normal form. That is,
each sequent $\Gamma\sequent \Delta$ in~$\pi$
has the form
\[
\bigvee \Pi_1, \dots, \bigvee \Pi_k
    \sequent \bigvee \Lambda_1, \dots, \bigvee \Lambda_\ell
\]
where each $\Pi_i$ and $\Lambda_j$ is a multiset of
$\DT$-formulas.  We shall prove that there
is a polynomial size $\DT$ derivation~$\pi^\prime$ of the
sequent
\begin{equation}\label{eq:flattenDelta}
\sequent \Lambda_1, \dots, \Lambda_\ell
\end{equation}
from the extra hypotheses $\sequent \Pi_i$. The proof is by induction
on the number of lines in the proof~$\pi$.
If $\pi$ is
just an axiom, then this is trivial.  Otherwise the argument splits
into cases depending on the final inference of~$\pi$.

For a more compact notation, we write $\mathcal{F}(\Delta)$ to
denote the succedent in~(\ref{eq:flattenDelta}) (``$\cal F$'' for
``flatten'').  And we write $\mathcal{H}(\Gamma)$ to denote the
set of sequents $\sequent \Lambda_i$ (``{$\cal H$}'' for ``hypotheses'').

If $\pi$ ends with a weakening or contraction inference,
the argument is essentially
trivial. For instance, if $\pi$ ends with a \emph{c-l} inference
\begin{center}
\Axiom$A , A , \Gamma \fCenter \Delta$
\LeftLabel{\emph{c-l:}}
\UnaryInf$A, \Gamma \fCenter \Delta$
\DisplayProof
\end{center}
then the induction hypothesis gives a $\LDT$ proof~$\pi_0^\prime$
of $\mathcal{F}(\Delta)$ from the hypotheses $\mathcal{H}(A, A, \Gamma)$.
But $\mathcal{H}(A, A, \Gamma)$ is equal to $\mathcal{H}(A, \Gamma)$,
we can just take $\pi^\prime$ to be~$\pi_0^\prime$.  The case
where $\pi$ ends with a \emph{w-l} inference
is handled similarly, since $\mathcal{H}(A, \Gamma)$ is a superset
of $\mathcal{H}(\Gamma)$.  If $\pi$ ends with a \emph{c-r} inference or
a \emph{w-r} inferences, we form $\pi^\prime$
by adding the same kind of inference to the end of the
$\LDT$ deduction~$\pi_0^\prime$ given by the induction hypothesis.

Suppose the final inference of~$\pi$ is a cut inference
\begin{center}
\Axiom$\Gamma \fCenter \Delta, A$
\Axiom$A, \Gamma \fCenter \Delta$
\LeftLabel{\emph{cut:}}
\BinaryInf$\Gamma \fCenter \Delta$
\DisplayProof
\end{center}
The cut formula~$A$ is an $\NDT$ formula, hence it is of the form
$\bigvee \Lambda$ for some cedent $\Lambda$ of $\DT$ formulas,
and ${\cal F}(A) = \Lambda$.

The two upper sequents of the cut have (disjoint since tree-like) $\TreeLNDT$ proofs
$\pi_0$ and~$\pi_1$.
The induction hypothesis gives an $\LDT$ proof~$\pi_0^\prime$
of the sequent $\sequent \mathcal{F}(\Delta), \Lambda$ from
the hypotheses $\mathcal{H}(\Gamma)$
and an $\LDT$ proof~$\pi_1^\prime$ of $\sequent \mathcal{F}(\Delta)$ from
the hypotheses $\sequent \Lambda$ and $\mathcal{H}(\Gamma)$.  We modify
$\pi_1^\prime$ to form a new $\LDT$ derivation, denoted
$\pi_1^\prime \srt \mathcal{F}(\Delta)$,
which is formed from $\pi_1^\prime$ by replacing each
sequent $\Pi \sequent \Xi$ in~$\pi_1^\prime$ with
$\Pi \sequent \Xi, \mathcal{F}(\Delta)$, and then
fixing up initial sequents to be validly derived by adding weakening inferences as needed.
This forms $\pi_1^\prime \srt \mathcal{F}(\Delta)$ as a $\LDT$-proof of
$\sequent \mathcal{F}(\Delta), \mathcal{F}(\Delta)$ from the hypotheses {$\cal H$}
and $\sequent \mathcal{F}(\Delta), \Lambda$.  We form the desired proof
$\pi^\prime$ by concatenating $\pi_0^\prime$ and
$\pi_1^\prime \srt \mathcal{F}(\Delta)$ and concluding with contraction inferences:
\begin{center}
\alwaysNoLine
\AxiomC{$\mathcal{H}(\Gamma)$}
\def\extraVskip{3pt}
\UnaryInfC{$\proofdotsL{\pi_0^\prime}$}
\UnaryInf$\fCenter \mathcal{F}(\Delta), \Lambda$
\UnaryInfC{$\proofdotsL{\pi_1^\prime \srt \mathcal{F}(\Delta)}$}
\UnaryInf$\fCenter \mathcal{F}(\Delta), \mathcal{F}(\Delta)$
\LeftLabel{\emph{c-r:}}
\doubleLine
\UnaryInf$\fCenter \mathcal{F}(\Delta)$
\DisplayProof
\end{center}
This yields $\pi^\prime$ as a polynomial size $\LDT$ proof
of $\sequent \mathcal{F}(\Delta)$ from the hypotheses~$\cal H$.

Now suppose the final inference of~$\pi$ is an $\lorRight$ inference
\begin{center}
\Axiom$\Gamma \fCenter \Delta, A, B$
\LeftLabel{\emph{$\lorRight$:}}
\UnaryInf$\Gamma \fCenter \Delta, A\lor B$
\DisplayProof
\end{center}
The $\NDT$ formulas $A$ and $B$ are equal
to $\bigvee \Pi$ and $\bigvee \Lambda$ where
$\Pi$ and~$\Lambda$ are cedents of $\DT$ formulas.
The induction hypothesis gives an $\LDT$ proof~$\pi_0^\prime$
of $\sequent \mathcal{F}(\Delta), \Pi, \Lambda$ from the hypotheses $\mathcal{H}(\Gamma)$.
The desired proof~$\pi^\prime$ is just equal to~$\pi_0$.

Now suppose the final inference of~$\pi$ is an $\lorLeft$ inference
\begin{center}
\Axiom$A, \Gamma \fCenter \Delta$
\Axiom$B, \Gamma \fCenter \Delta$
\LeftLabel{\emph{$\lorLeft$:}}
\BinaryInf$ A\lor B, \Gamma \fCenter \Delta$
\DisplayProof
\end{center}
The $\NDT$ formulas $A$ and $B$ are again equal
to $\bigvee \Pi$ and $\bigvee \Lambda$. The induction
hypothesis gives an $\LDT$ proof~$\pi_0^\prime$ of
$\sequent \mathcal{F}(\Delta)$ from the hypotheses
$\sequent \Pi$ and $\cal H(\Gamma)$,
and gives an $\LDT$ proof~$\pi_1^\prime$ of
$\sequent \mathcal{F}(\Delta)$ from the hypotheses
$\sequent \Lambda$ and~${\cal H}(\Gamma)$. We must produce an $\LDT$ proof~$\pi^\prime$
of $\sequent \mathcal{F}(\Delta)$ from the hypotheses
$\sequent\penalty 10000 \Pi, \Lambda$ and~$\mathcal{H}(\Gamma)$.  We form
$\pi_0^\prime \srt \Lambda$ by adding $\Lambda$ to the
antecedent of each sequent in~$\pi_0^\prime$, and then fixing
up all initial sequents with weakening inferences, except leaving the
initial sequents $\sequent \Pi, \Lambda$ as is. This makes
$\pi_0\srt\Lambda$ an $\LDT$ derivation of
$\sequent \mathcal{F}(\Delta), \Lambda$ from the hypotheses
$\sequent \Pi, \Lambda$ and $\mathcal{H}(\Gamma)$.  We similarly form
$\pi_1^\prime \srt \mathcal{F}(\Delta)$ to be a $\LDT$ proof
of $\sequent \mathcal{F}(\Delta), \mathcal{F}(\Delta)$ from the
hypotheses $\sequent \mathcal{F}(\Delta), \Lambda$ and $\mathcal{H}(\Gamma)$.
Putting these together as:
\begin{center}
\alwaysNoLine
\AxiomC{$\sequent \Pi,\Lambda$  \quad ${\cal H}(\Gamma)$}
\kernHyps{-8pt}
\def\extraVskip{3pt}
\UnaryInfC{$\proofdotsL{\pi_0^\prime \srt \Lambda}$}
\UnaryInf$\fCenter \mathcal{F}(\Delta), \Lambda$
\UnaryInfC{$\proofdotsL{\pi_1^\prime \srt \mathcal{F}(\Delta)}$}
\UnaryInf$\fCenter \mathcal{F}(\Delta), \mathcal{F}(\Delta)$
\LeftLabel{\emph{c-r:}}
\doubleLine
\UnaryInf$\fCenter \mathcal{F}(\Delta)$
\DisplayProof
\end{center}
forms the desired $\LDT$ proof of $\sequent \mathcal{F}(\Delta)$
from the hypotheses $\sequent \Pi,\Lambda$ and~${\cal H}(\Gamma)$.

Now suppose the final inference of~$\pi$ is a \emph{dec-r} inference
\begin{center}
\Axiom$\Gamma \fCenter \Delta, A, p$
\Axiom$p, \Gamma \fCenter \Delta, B$
\LeftLabel{\emph{dec-r:}}
\BinaryInf$\Gamma \fCenter \Delta, ApB$
\DisplayProof
\end{center}
$A$ and $B$ are $\DT$ formulas. The induction
hypothesis gives an $\LDT$ proof~$\pi_0^\prime$
of $\sequent \mathcal{F}(\Delta), A, p$ from
the hypotheses $\mathcal{H}(\Gamma)$ and an $\LDT$ proof~$\pi_1^\prime$
of $\sequent \mathcal{F}(\Delta), B$ from the hypotheses
$\sequent p$ and $\mathcal{H}(\Gamma)$.
We form an $\LDT$ proof~$p \srt \pi_1^\prime$ by adding
$p$ to each antecedent, replacing the hypothesis $\sequent p$
with the axiom $p\sequent p$, and adding weakenings to fix
up the other initial sequents. The desired
$\LDT$ proof~$\pi^\prime$ is formed as:
\begin{center}
\alwaysNoLine
\def\extraVskip{3pt}
\AxiomC{$\mathcal{H}(\Gamma)$}
\UnaryInfC{$\proofdotsL{\pi_0^\prime}$}
\UnaryInf$\fCenter \mathcal{F}(\Delta), A, p$
\AxiomC{$\mathcal{H}(\Gamma)$}
\UnaryInfC{$\proofdotsL{p \srt \pi_1^\prime}$}
\UnaryInf$p \fCenter \mathcal{F}(\Delta), B$
\singleLine
\LeftLabel{\emph{dec-r:}}
\BinaryInf$\fCenter \mathcal{F}(\Delta), ApB$
\DisplayProof
\end{center}

Finally  suppose the final inference of~$\pi$ is a \emph{dec-l} inference
\begin{center}
\Axiom$A,\Gamma \fCenter \Delta, p$
\Axiom$p, B, \Gamma \fCenter \Delta$
\LeftLabel{\emph{dec-l:}}
\BinaryInf$ApB, \Gamma \fCenter \Delta$
\DisplayProof
\end{center}
where $A$ and $B$ are again $\DT$ formulas, and the induction
hypothesis gives an $\LDT$ proof~$\pi_0^\prime$ of
$\sequent \mathcal{F}(\Delta), p$ from the hypotheses
$\sequent A$ and $\mathcal{H}(\Gamma)$ and an $\LDT$ proof~$\pi_1^\prime$
of $\sequent \mathcal{F}(\Delta)$ from the hypotheses
$\sequent p$ and $\sequent B$ and $\mathcal{H}(\Gamma)$.
We need to form an $\LDT$ proof of $\sequent \mathcal{F}(\Delta)$ from the
hypothesis $\sequent ApB$ and $\mathcal{H}(\Gamma)$.
From Proposition \ref{prop:LDTidentity}(f,g), there
are short $\LDT$ proofs of $ApB\sequent A,p$ and
$ApB,p \sequent B$. Similarly to the previous cases,
we form an $\LDT$ proof $\pi_0^\prime \srt p$ of
$\sequent \mathcal{F}(\Delta), p$ from the hypotheses
$\sequent A, p$ and $\mathcal{H}(\Gamma)$. We also
form an $\LDT$ proof~$p \srt \pi_1^\prime$ of
$p \sequent \mathcal{F}(\Delta)$ from the hypotheses
$p \sequent B$ and $\mathcal{H}(\Delta)$. Combining all these
with cuts gives the desired $\LDT$ proof~$\pi$ as:
\begin{center}
\def\extraVskip{3pt}

\Axiom$\fCenter ApB$
\AxiomC{$\proofdotsL{\parbox[b]{2cm}{\raggedright Prop.\\ \ref{prop:LDTidentity}(f)}}$}
\noLine
\kernHyps{-10pt}
\UnaryInf$ApB\fCenter A, p$
\LeftLabel{\emph{cut:}}
\BinaryInf$\fCenter A,p$
\noLine
\UnaryInfC{$\proofdotsL{\pi_0^\prime \srt p}$}
\noLine
\UnaryInf$\fCenter \mathcal{F}(\Delta), A, p$

\Axiom$\fCenter ApB$
\AxiomC{$\proofdotsL{\parbox[b]{2cm}{\raggedright Prop.\\ \ref{prop:LDTidentity}(g)}}$}
\noLine
\kernHyps{-10pt}
\UnaryInf$ApB, p \fCenter B$
\LeftLabel{\emph{cut:}}
\BinaryInf$p \fCenter B$
\noLine
\UnaryInfC{$\proofdotsL{p \srt \pi_1^\prime}$}
\noLine
\UnaryInf$p, \fCenter \mathcal{F}(\Delta)$

\LeftLabel{\emph{cut:}}
\BinaryInf$\fCenter\mathcal{F}(\Delta)$
\DisplayProof
\end{center}

It is not hard to verify that proof $\pi^\prime$ is constructible from~$\pi$ in polynomial time.
That completes the proof of Theorem~\ref{thm:tree-lndt-equiv-ldt}.
\end{proof}

\subsection{Equivalence of \texorpdfstring{$\LNDT$}{LNDT} and \texorpdfstring{$\twoLK$}{2-LK}}
\label{sec:2LK_LNDT}

A Boolean formula is \emph{depth two} if it is depth one, or
if it is a conjunction of clauses or a disjunction of terms.
$\twoLK$ is the fragment of $\LK$ in which all formulas
appearing in sequents are depth two formulas. $\TreetwoLK$ is the
same system with the restriction that proofs are tree-like.

\begin{theorem}\label{thm:2LK_LNDT}
$\LNDT$ and $\twoLK$ are polynomially equivalent.
$\TreeLNDT$ and $\TreetwoLK$ are polynomially equivalent.
\end{theorem}
The equivalence between $\LNDT$ and $\twoLK$ is even
stronger than is required by Definition~\ref{def:simulations}.
In fact, \emph{any} $\LNDT$ proof can be faithfully
translated into a $\twoLK$ proof.  For the converse,
we sketch below how any $\twoLK$ proof in which the final sequent is
contains only disjunctions of conjunctions
can be faithfully translated to a $\LNDT$ proof. This means essentially
that \emph{any} $\twoLK$ proof can be faithfully translated to a $\LNDT$
proof, since any conjunctions of disjunctions can be moved
to the other side of the sequent where they become
disjunctions of conjunctions.

\begin{proof}(Sketch)
Suppose $\pi$ is a $\LNDT$ proof. By Theorem~\ref{thm:LNDT_NF_wlog},
every formula in~$\pi$ may be assumed to be a normal form $\NDT$ formula. To convert
$\pi$ to a $\twoLK$ proof~$\pi^\prime$, we first replace every formula $\bigvee A_i$
in~$\pi$ with the depth two Boolean formula
$\bigvee \Tms(A_i)$. Axioms and contraction, weakening,
cut and $\lor$ inferences in~$\pi$ remain valid inferences
in~$\pi^\prime$. Decision rules \emph{dec-l} and \emph{dec-r}
in~$\pi$ are easily fixed to be valid derivation in $\pi^\prime$
using axioms $\overline p,p \sequent$ and $\sequent \overline p, p$,
cuts on~$p$, and $\land$ and $\lor$ inferences.  The resulting
$\twoLK$ proof~$\pi^\prime$ has size linearly bounded by the size
of~$\pi$. In addition, if $\pi$ is tree-like, then so is~$\pi^\prime$.

Conversely, suppose $\pi$ is a $\twoLK$ proof, and that every formula
in the conclusion of~$\pi$ is a disjunction of conjunctions of literals.
We may assume w.l.o.g.\ that every
formula in~$\pi$ is a disjunction of conjunctions of literals, since
any conjunction of disjunctions can be negated and moved to the other side of
the cedent as a disjunction of conjunctions.
We thus can transform $\pi$ into $\pi^\prime$
by replacing every formula $\bigvee A_i$ in~$\pi$,
where the $A_i$'s are conjunctions of literals,
with the $\NDT$ formula $\bigvee \Conj(A_i)$.
The axioms and the contraction, weakening, cut and $\lor$
inferences in~$\pi$ remain valid after this transformation.
The $\land$ rules in~$\pi$ can be fixed to be valid derivations
in~$\pi$ using the derivations of Proposition~\ref{prop:ConjDisj}(a,c,e)
and cuts on formulas $\Conj(\vec p)$ and $\Conj(\vec q)$ for
$\vec p$ and~$\vec q$ vectors of literals.
\end{proof}

\section{Proof systems for branching programs}\label{sec:eLDT_eLNDT}

\subsection{Formulas and proofs with extension variables}

We now describe the propositional proof systems $\eLDT$ and $\eLNDT$ which
reason about deterministic and nondeterministic branching programs.\footnote{%
These systems could equally well be called LBP and LNBP, using ``BP'' for
``branching programs'', but the
notations $\eLDT$ and $\eLNDT$ indicate that branching programs are
represented with decision trees incorporating extension variables.}
Formulas can now include extension variables, which will be denoted
by the letter~$e$, or with a subscript as $e_1$, $e_2$, etc..
It is important that the extension variables~$e$ are new variables
that are distinct from the variables underlying literals~$p$.

The purpose of extension variables is to serve as abbreviations for
more complex formulas. Thus, proofs that use extension variables will
be accompanied by a set of extension axioms $\{e_i\liff A_i\}_{i<n}$,
where each formula~$A_i$ may use any literals~$p$ but is restricted
to use only the extension variables $e_j$ for $j<i$. The intent is
that $e_i$ is an abbreviation for the formula~$A_i$.

\begin{definition}\label{def:eDTformula}
    The \emph{extended decision tree} formulas, or $\eDT$ formulas for short,
    are inductively defined
    \begin{integerenumerate}
        \item Any literal $p$ is an $\eDT$ formula.
        \item Any extension variable $e$ is an $\eDT$ formula.
        \item If $A$ and $B$ are $\eDT$ formulas and $p$ is a literal, then
            $(A p B)$ is a $\DT$ formula.
    \end{integerenumerate}
\end{definition}
In particular, a decision
            literal~$p$ in a formula $ApB$ is \emph{not} allowed to be an extension variable.
The intuition is that the extension variables may `name' nodes in a branching program.
\begin{definition}\label{def:eNDTformula}
The \emph{extended nondeterministic decision tree} formulas,
or $\eNDT$ formulas for short,
are inductively defined by the closure conditions (1)-(3) above
(with ``$\eDT$'' replaced with ``$\eNDT$'') and:
\begin{integerenumerate}[start=4]
    \item If $A$ and $B$ are $\eNDT$ formulas,
        then $(A\lor B)$ is an $\eNDT$ formula.
\end{integerenumerate}
\end{definition}

\begin{definition}
    The \emph{extended Boolean formulas} are defined inductively by
    \begin{integerenumerate}
        \item Any literal $p$ is a extended Boolean formula.
        \item Any extension variable~$e$
            is an extended Boolean formula.
        \item If $A$ and $B$ are extended Boolean formulas, then
            so are $(A\lor B)$ and $(A \land B)$.
    \end{integerenumerate}
\end{definition}

The notation $\{e_i \liff A_i\}_{i<n}$ is used to indicate that
$e_0,\dots, e_{n-1}$ are extension variables and that the only
extension variables allowed to appear in~$A_i$ are
$e_0,\dots,e_{i-1}$.  The sequents
\[
e_i \sequent A_i \qquad\hbox{and}\qquad A_i\sequent e_i
\]
are called the \emph{extension axioms}.

The $\eDT$, $\eNDT$ and $\eLK$ formulas have truth semantics
only relative to a set of extension axioms $\{e_i \liff A_i\}_{i<n}$.
Namely, for $\alpha$ a truth assignment, the definition of truth
is extended by setting $\alpha(e_i) = \alpha(A_i)$.

\begin{definition}
An $\eLDT$ proof is a pair $(\pi, \{e_i \liff A_i\}_{i<n})$
where each $A_i$ is an $\eDT$ formula, all formulas in~$\pi$ are
$\eDT$ formulas, and the permitted initial sequents and rules of
$\DT$ plus the
extension axioms of $\{e_i \liff A_i\}_{i<n}$ are allowed
as initial sequents in~$\pi$.

The $\eLNDT$ proofs are defined similarly, but with $\eLNDT$ formulas~$A_i$
and using the $\eLNDT$ inference rules.  Similarly, $\eLK$ proofs are defined
by letting the $A_i$ be $\eLK$ formulas and using the $\LK$ inference rules.
\end{definition}

Clearly the $\eLK$ proof system is equivalent to
the usual extended Frege proof system: in conjunction
with a set of extension axioms, an extended Boolean
formula represents a Boolean circuit over the de Morgan
connectives $\land, \lor, \lnot$.

Note that all formulas in an $\eLDT$, $\eLNDT$ or $\eLK$ proof are based on the
a single set of extension axioms $\{e_i \liff A_i\}_{i<n}$.

\smallskip

Let us discuss how the extended formulas we have introduced may be used to represent bona fide branching programs.
A (deterministic) branching program is a directed acyclic graph~$G$ such that (a)~$G$ has a
unique source node, (b)~sink nodes in~$G$ are labelled
with either $0$ or $1$, (c)~all other nodes are labelled
with a literal~$p$ and have two outgoing edges,
one labelled $0$ and the other $1$. A deterministic
branching program~$G$ can be converted into an equivalent
$\eDT$ formula
with associated extension
axioms $\{e_i \liff A_i\}_{i < n}$ by introducing
an extension variable~$e_i$ for every 
internal 
node in the
branching program. 
Conversely,
as is described in more detail below, any
$\eDT$ formula~$A$ with extension axioms
$\{e_i \liff A_i\}_{i < n}$ can be straightforwardly
transformed into a linear size deterministic branching program.
For this, the nodes in the branching program correspond to
the extension variables $e_i$ and the subformulas of the
formulas~$A_i$.

Nondeterministic branching programs are defined similarly
to deterministic branching programs, but further allowing the
internal nodes of~$G$ to be labelled with ``$\lor$'' as well as literals (in this case the labelling of its outgoing edges is omitted). The semantics is that an $\lor$-node
is accepting provided at least one of its children
is accepting. It is straightforward to convert
a nondeterministic branching program into an
$\eLNDT$ formula with associated extension axioms,
and vice versa.

A similar construction yields the well-known fact
that extended Boolean formulas are as expressive
as Boolean circuits.

\begin{example}
Consider the following branching program, which returns $1$ just if at least two out of the four input variables $w,x,y,z$ are $1$.
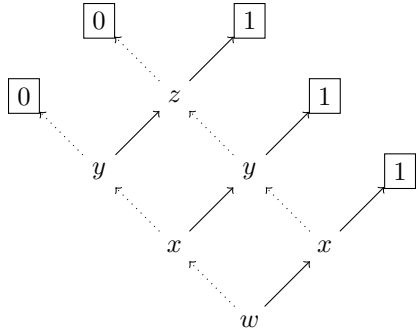
\begin{equation*}
    \begin{tikzpicture}
 \node (0) at (0,0) {$w$};
 \node (1-0) at (-1,1) {$x$};
 \node (1-1) at (1,1) {$x$};
 \node (2-0) at (-2,2) {$y$};
 \node (2-1) at (0,2) {$y$};
 \node[rectangle,draw] (2-2) at (2,2) {$1$};
 \node[rectangle,draw] (3-0) at (-3,3) {$0$};
 \node (3-1) at (-1,3) {$z$};
 \node[rectangle,draw] (3-2) at (1,3) {$1$};
 \node[rectangle,draw] (4-0) at (-2,4) {$0$};
 \node[rectangle,draw] (4-1) at (0,4) {$1$};
 \draw[->,dotted] (0) to (1-0);
 \draw[->] (0) to (1-1);
 \draw[->,dotted] (1-0) to (2-0);
 \draw[->] (1-0) to (2-1);
 \draw[->,dotted] (1-1) to (2-1);
 \draw[->] (1-1) to (2-2);
 \draw[->,dotted] (2-0) to (3-0);
 \draw[->] (2-0) to (3-1);
 \draw[->,dotted] (2-1) to (3-1);
 \draw[->] (2-1) to (3-2);
 \draw[->,dotted] (3-1) to (4-0);
 \draw[->] (3-1) to (4-1);
\end{tikzpicture}\label{fig:bp-to-edt}
\end{equation*}
Edges labelled with $0$ are here dotted (and always left outgoing) while edges labelled $1$ are here solid (and always right outgoing).
In this particular case, the branching program is \emph{ordered} (or an \emph{OBDD}), i.e.\ variables occur in the same order on each branch. The program also happens to compute a monotone Boolean function.

To express the branching program above in $\eLDT$, we introduce extension variables for each inner node of the program as follows. Write $e_{ij}$ for the $j$th node of the $i$th layer, where $i,j$ ranging from $0$ onwards, and introduce the following extension axioms:\footnote{Formally, we are writing $0$ and $1$ as shorthand for $ p p \bar p$ and $\bar p p p $ respectively, for some/any literal $p$.}
\[
\begin{array}{rcl}
    e_{10} & \liff & e_{20}xe_{21} \\
    e_{11} & \liff & e_{21}x1 \\
    e_{20} & \liff & 0ye_{31} \\
    e_{21} & \liff & e_{31}y1 \\
    e_{31} & \liff & 0z1
\end{array}
\]
Now the branching program is represented as the eDT formula $e_{10}we_{11}$.
Notice that the orderedness of the branching program is reflected in its $\eLDT$ representation:
writing $(x_0,x_1,x_2,x_3) $ for $(w,x,y,z)$, we have that $x_i$ is the root of the formula that any $e_{ij}$ abbreviates.

Other representations of this branching program are possible,
for instance by renaming the extension variables or by partially unwinding the graph.
In both these two latter cases, the $\eDT$ representation obtained
will be provably equivalent to the one above, by polynomial-size proofs
in $\eLDT$, by virtue of Lemma~\ref{lem:extNamesEquiv} later.
\end{example}

\subsection{Foundational issues}

The fact that extension variables cannot be used as
decision literals is a significant limitation on the
expressiveness of $\DT$ formulas. Recall for instance that
the conjunction of $p_1$ and~$p_2$ can be expressed
with the $\DT$ formula $\Conj(p_1, p_2)$, namely
$(p_1 p_1 p_2)$. However, it is not permitted
to form $(e_1 e_1 e_2)$; in fact, it is not possible
to express the conjunction $e_1 \land e_2$ without
taking the extension axioms defining $e_1$ and $e_2$ into
account.  In fact, if we could write the conjunction of $e_1$
and~$e_2$ by a generic formula $A(e_1, e_1)$, then we could
introduce a new extension variable representing
$A(e_1, e_2)$.  This would imply that $\eDT$ formulas are as expressive
as extended Boolean formulas; in other words, that
deterministic branching programs would be as expressive as
Boolean circuits.  This is a non-uniform analogue of
$\logspace = \Ptime$ (i.e., log-space equals polynomial time), and of course is an
open question.

Nonetheless, for any given extension variables $e$ and~$e^\prime$,
there is a formula $\AndDT(e,e^\prime)$ expressing the conjunction
of $e$ and~$e^\prime$ by changing the underlying set of extension axioms.  The intuition is that we start with the
branching program~$G$ for~$e$, but now with sink nodes
labelled with $0$ or~$1$ instead of with variables. To form the branching
program for $e\land e^\prime$, we take (an isomorphic copy) of the
branching program~$G^\prime$ for~$e^\prime$, and modify
$G$ by replacing each sink node labelled
with~$1$ with the source node of~$G^\prime$ (in other words,
each edge directed into a sink~``$1$'' is modified to instead
point to the root of~$G^\prime$).

More formally, suppose $A$ and $B$ are $\eDT$ formulas defined
over a set of extension axioms $\{e_i \liff A_i\}_{i < n}$; we
wish to construct an $\eDT$ formula $\AndDT(A,B)$.  (Exactly the
same construction forms an $\eNDT$ formula $\AndNDT(A,B)$ from
$\eNDT$ formulas $A$ and~$B$.) We would wish to define $C[1/B]$
to be the result of replacing every ``$1$'' in~$C$ with~$B$, but
of course, ``$1$'' is not a permitted atom. Instead, we note
that every atomic formula~$p$ in~$C$ is equivalent to
$(p p p)$ and to $(p p 1)$. Likewise, each atomic formula~$p$
is equivalent to $(0 p p )$.

\begin{definition}
Let $C$ be an $\eDT$ or $\eNDT$
formula. $C[0/B]$
is the formula obtained by
replacing (in parallel) each occurrence of a literal~$p$ as a leaf in~$C$
with the formula $(B \, p \, p)$.
Similarly,
$C[1/B]$ is the formula obtained by replacing
each occurrence of a literal~$p$ as a leaf in~$C$
with the formula $(p \, p\, B)$.
\end{definition}
The point of $C[0/B]$ is that $(B \, p \, p)$
evaluates to 1 if $p$~is true, and to~$B$  otherwise.
Thus, the intent is that $C[0/B]$ is equivalent $C\lor B$.
Likewise, we want $C[1/B]$ to be equivalent $C\land B$.
However, these equivalences hold only if the substitutions are
applied not just in~$C$ but instead
throughout the definitions of the extension axioms used in~$C$.
This is done with the following definition.

\begin{definition}\label{def:andOrDT}
Let $\cal A$ be a set of extension axioms $\{e_i \liff A_i\}_{i < n}$.
Another set of extension axioms $\mathcal{A}[1/B]$ is defined as follows.
First, let $\{e_i^\prime\}_i$ be
a set of \emph{new} extension variables.
Define $A_i[\veceprime/\vec e]$ to be the result of
replacing each $e_j$ in~$A_i$ with $e_j^\prime$. Let $A_i^\prime$
be $(A_i[\veceprime/\vec e])[1/B]$. Then $\mathcal{A}[1/B]$ is
the set of extension axioms
$\{e^\prime_i \liff A^\prime_i\}_{i < n} \cup \mathcal{A}$.
The set $\mathcal{A}[0/B]$ is defined similarly: letting
$\vecepprime$ be another set of new extension variables,
defining $A_i^\pprime$ to be $(A_i[\vecepprime/\vec e])[0/B]$, and
letting $\mathcal{A}[0/B]$ be the set of extension axioms
$\{e^\pprime_i \liff A^\pprime_i\}_{i < n} \cup \mathcal{A}$.

Finally, if $A$ and $B$ are $\eDT$ or $\eNDT$ formulas
defined using extension axioms~$\cal A$, then
$\AndDT(A,B)$ is by definition $A[1/B]$ relative to the extension
axioms $\mathcal{A}[1/B]$.  The formula $\OrDT(A,B)$ for disjunction is defined similarly,
namely, it is equal to $A[0/B]$ relative to the extension
axioms $\mathcal{A}[0/B]$.
\end{definition}

Note the two formulas $\AndDT(A,B)$ and $\OrDT(A,B)$ introduced
\emph{different} sets of new extension variables. This allows us
to use both $\AndDT(A,B)$ and $\OrDT(A,B)$ without any clashes
between extension variables.
More generally, we will adopt the
convention that the new extension variables are uniquely determined by the
formula being constructed. In other words, for instance, $e_i^\prime$ could have instead been
designated $e_{i,(A\land B)}$. When measuring proof size, we also need to count the sizes
of the subscripts on the extension variables. This clearly however only increases proof size
polynomially.

There are two other sources of growth of size in forming $\AndDT(A,B)$ and $\OrDT(A,B)$.
The first is that formula sizes increase since copies of $B$ is substituted
in at many places in $A$ and $\cal A$: this potentially gives a quadratic blowup in
proof size. We avoid this quadratic blowup in proof size, by always taking $B$ to
be a single variable (namely, an extension variable). The construction
of $\AndDT(A,B)$ or $\OrDT(A,B)$
also introduces many new extension variables, namely it potentially doubles the number of
variables. To control this, we will ensure that the constructions of $\AndDT(\cdot, \cdot)$
and $\OrDT(\cdot, \cdot)$ are nested only logarithmically.

\begin{example}
Consider the formula $\AndDT(p_1, \AndDT(p_2,p_3))$, which
is a translation of the Boolean formula $p_1 \land( p_2 \land p_3)$ to
a $\DT$ formula.  To form $\AndDT(p_2,p_3)$, start with $(p_2 p_2 1)$
and substitute $p_3$ for ``$1$'', to obtain $(p_2 p_2 p_3)$. Then
$\AndDT( p_1, \AndDT(p_2,p_3) )$ is obtained by forming
$(p_1 p_1 1)$ and replacing ``$1$'' with $\AndDT(p_2,p_3)$ to obtain
$(p_1 p_1 ( p_2 p_2 p_3)$. It is also the same as $\Conj(p_1,p_2,p_3)$.  A similar construction shows
that $\OrDT(p_1, \OrDT(p_2, p_3)$ is equal to $((p_3 p_2 p_2)p_1 p_1)$.
This is a translation of the Boolean formula $p_1 \lor (p_2\lor p_3)$
to a $\DT$ formula, and is equal to $\Disj(p_1, p_2, p_3)$.
\end{example}

\begin{example}
	\label{example:bp-to-edt-formula-th42}
Let $A$ be the formula $(p_1 p_2 ( e_1 p_3 e_2))$ and
$B$~be the formula $(q_1 q_2 e_2)$ in the context of
the extension axioms~$\cal A$
\begin{equation}\label{eq:calAexample}
  e_1 \liff (r_1 \overline{r_2} e_2) \qquad\qquad
  e_2 \liff (\overline{s_1} s_2 s_3),
\end{equation}
where $p_i,q_i,r_i,s_i$ are literals. The formula $A[0/B]$ is formed
as follows.  First $\mathcal{A}(\veceprime/\vec e)$ equals
\[
  e_1^\prime \liff (r_1 \overline{r_2} e_2^\prime) \qquad\qquad
  e_2^\prime \liff (\overline{s_1} s_2 s_3)
\]
Then $\mathcal{A}[0/B]$ contains the extension axioms
of $\cal A$ as shown in (\ref{eq:calAexample}) plus the
extension axioms
\[
  e_1^\prime \liff ((B r_1 r_1) \overline{r_2} e_2^\prime) \qquad\qquad
  e_2^\prime \liff ((B\, \overline{s_1}\, \overline{s_1}) s_2 (B s_3 s_3)).
\]
Finally, $A[0/B]$ is the $\DT$ formula
$((B p_1 p_1) p_2 (e_1^\prime p e_2^\prime))$,
namely,
\[
(((q_1 q_2 e_2) p_1 p_1) p_2 (e_1^\prime p e_2^\prime)),
\]
relative to the four extension axioms in $\mathcal{A}[0/B]$.
\end{example}

\subsection{Truth conditions and renaming of extension variables}
We show that, despite the delicate renaming of variables required for notions such as $A[0/B]$ and $\AndDT(A,B)$, for $\DT$ (respectively $\NDT$) formulas $A,B$, we may nonetheless realise their basic truth conditions by small $\eLDT$ (respectively $\eLNDT$) proofs:

\begin{lemma}\label{lem:eDTAndOr}
Let $A$ and $B$ be $\eDT$ formulas (respectively, $\eNDT$ formulas)
relative to extensions axioms~$\cal A$. Then, the sequents {\rm (a)-(c)} below
have polynomial size, cut free $\eLDT$ proofs (respectively, $\eLNDT$ proofs)
relative to the extension axioms $\mathcal{A}[0/B]$. The same holds for
the sequents {\rm (d)-(f)} relative to $\mathcal{A}[1/B]$.
\begin{multicols}{2}
\begin{alphenumerate}[nolistsep]
\item $B \sequent A[0/B]$
\item $A \sequent A[0/B]$
\item $A[0/B] \sequent A, B$
\item $A[1/B] \sequent B$
\item $A[1/B] \sequent A$
\item $A, B \sequent A[1/B]$
\end{alphenumerate}
\end{multicols}
\end{lemma}
\begin{proof}
[Proof sketch]
    Parts (a)-(c) are  proved by showing inductively that if $C$ is a
    subformula of $A[0/B]$ or a subformula of any $A_i^\prime$ in~$\mathcal{A}[0/B]$,
    then $C \sequent A, B$ and $B \sequent C$ and $A \sequent C$ have
    short $\eLDT$ (resp., $\eLNDT$) proofs. The base cases are just the
    cases where $C$ is is the form $(B \, p \, p)$. The inductive cases
    are trivial.
    A similar argument proves cases (d)-(f).
\end{proof}

The proofs of Lemma~\ref{lem:eDTAndOr} seem to be inherently dag-like, and we
do not know if the lemma holds for $\TreeELDT$.

As discussed above, we assume that the choice of
new extension variables $\veceprime$ or $\vecepprime$ depends
explicitly on what formula $\AndDT(A,B)$ and $\OrDT(A,B)$ is being formed.
In other words, each $e^\prime_i$ or $e_i^\pprime$
is a variable $e_{i,\AndDT(A,B)}$ or
$e_{i,\OrDT(A,B)}$.
In the proof of Theorem~\ref{thm:ELDT_LK}, this means
that the translations of distinct occurrences of the same
Boolean formula use the same extension variables.
However, this is not strictly
necessary, as $\eLDT$ can prove the equivalence of formulas after a change
in extension variables:
\begin{lemma}\label{lem:extNamesEquiv}
Suppose $A$ is a $\DT$ formula w.r.t.\ extension axioms
$\mathcal{A} = \{ e_i \liff A_i \}_i$, and that the extension
variables $\vec f$ are distinct from the extension variables~$\vec e$.
Let $B$ equal $A[\vec f/ \vec e]$ w.r.t.\ the extension axioms
$\mathcal{B} = \{ f_i \liff A_i[\vec{f} / \vec{e}]\}_i$. Then
$\eLDT$ has a polynomial size, cut free (dag-like) proofs
of $A\sequent B$ and $B \sequent A$ relative to the extension
axioms $\mathcal{A} \cup \mathcal{B}$.
\end{lemma}
Lemma~\ref{lem:extNamesEquiv} has a straightforward proof that proceeds inductively
through all subformulas of the formulas $A_i$ and $A$. \hfill $\Box$

\section{Simulations for \texorpdfstring{$\eLDT$}{eLDT}, \texorpdfstring{$\eLNDT$}{eLNDT}
    and \texorpdfstring{$\LK$}{LK}}\label{sec:eDTequivs}

\subsection{$\eLDT$ polynomially simulates $\LK$}\label{sec:ELDT_LK}

\begin{theorem}\label{thm:ELDT_LK}
$\eLDT$ polynomially simulates $\LK$. Hence, $\eLNDT$ also polynomially
simulates $\LK$.
\end{theorem}

The intuition behind this theorem is that the formulas in an $\LK$ proof
are Boolean formulas, and hence express $\NC^1$ properties, while $\DT$ proofs work
with $\DT$ formulas that express (nonuniform) logspace properties. Since
Boolean formula evaluation can be done in logspace, it is expected that
$\DT$ can directly simulate an $\LK$ proof.  This is indeed how the proof
goes, but it is complicated by the need to the $\AndDT$ and $\OrDT$
constructions.

\begin{proof}
Suppose $\pi$ is an $\LK$ proof of a sequent of Boolean formulas (possibly, but
not necessarily
of the form~(\ref{eq:fromLK})).
We wish to convert $\pi$ into
a $\eLDT$ proof. The main technique is to use the constructions
$\AndDT$ and $\OrDT$ of Definition~\ref{def:andOrDT} to convert
the Boolean formulas in~$\pi$ into $\DT$ formulas over extension axioms.
However, some care is
needed to ensure that the resulting $\DT$ formulas and extension axioms
are polynomial size.

For this, let $L(A)$ denote the \emph{leaf size} of the formula~$A$, 
namely the number of atomic subformulas of~$A$. The leaf size
$L(\mathcal{A})$ of a set of extension axioms is $\sum_i L(A_i)$.
A straightforward analysis shows that Definition~\ref{def:andOrDT}
constructs $\AndDT(A, B)$ to have leaf size $\le L(A) \cdot (L(B) + 1)$,
and $L(\mathcal{A}[1/B]$ to be $\le L(\mathcal{A}) \cdot (L(B) + 2)$. To avoid
too large formulas sizes, we will require that $L(B)=1$. When this
holds, we have $L(\AndDT(A,B)) \le 2 L(A)$ and
$L(\mathcal{A})[1/B] \le 3 L(\mathcal{A})$. The same size bounds hold
for $\OrDT(A, B)$ of course.

The \emph{height}
of a Boolean formula~$A$ is the height of the binary tree
corresponding to the formula~$A$.
Let's assume every formula in the endsequent of~$\pi$ has logarithmic
height. Then by
\cite{Reckhow:thesis,Buss:selfcons}, we
may assume w.l.o.g.\ that every formula in~$\pi$ has
height $O(\log |\pi|)$.\footnote{%
We can also assume without loss of generality that $\pi$ is
a tree-like proof. This, however, does not help form a
tree-like $\DT$ proof, since Lemma~\ref{lem:eDTAndOr} uses
dag-like proofs in an essential way.}
Each formula~$A$ in~$\pi$ is converted into
a $\DT$ formula $\DTx(A)$ with associated extension axioms $\mathcal{A}_A$
as defined next. The formula $\DTx(A)$ will always be either
a literal~$p$ or an extension variable $e_A$.
\begin{alphenumerate}
\item Suppose $A$ is a literal~$p$, then $\DTx(A)$ is just $p$, and
$\mathcal{A}_p$ is empty (no extension axioms).
\item If $A$ is $B\land C$, then let $\DTx(A)$ be the (new) extension
variable~$e_A$. Letting $\mathcal{A}^\prime$ be $\mathcal{A}_B \cup \mathcal{A}_C$,
set $\mathcal{A}_A$ equal to
$\mathcal{A}^\prime [1/C] \cup \{ e_A \liff \AndDT(\DTx(B),\DTx(C))\}$.
\item The case where $B\lor C$ is exactly the same, but with
$\mathcal{A}_A$ equal to
$\mathcal{A}^\prime [0/C] \cup \{ e_A \liff\penalty10000 \OrDT(\DTx(B),\DTx(C))\}$.
\end{alphenumerate}
Recall the convention that the new extension variables introduced in cases (b) and~(c)
depend uniquely on $A$. This implies that every occurrence of a given formula~$A$
in the proof~$\pi$ has the identical translation $\DTx(A)$.  Furthermore, the
formulas $\DTx(A)$ and $\DTx(B)$ share extension variable precisely to the
extent that they share subformulas.  More precisely, if $C$ is a subformula of~$A$,
then $\DTx(A)$ uses the extension variable~$e_C$ to denote the subformula~$C$,
using exactly the same extension axioms $\mathcal{A}_C$.

With these constructions, the $\LK$ proof~$\pi$ is translated
to a $\DT$ proof by replacing every (Boolean) formula~$A$ in~$\pi$ with
the $\DT$ formula $\DTx(A)$ and using
as extension axioms, the set $\bigcup_A \mathcal{A}_A$ where the
union is taken over all formulas~$A$ appearing in~$\pi$.
This yields $\pi^\prime$, and we claim
this can readily be fixed up to be a valid $\DT$ proof.
For instance, an $\lorRight$ in~$\pi$
\begin{prooftree}
\Axiom$\Gamma \fCenter \Delta, A$
\Axiom$\Gamma \fCenter \Delta, B$
\LeftLabel{\emph{$\lorRight$:}}
\BinaryInf$\Gamma \fCenter \Delta, A\land B$
\end{prooftree}
gets transformed to
\begin{prooftree}
\Axiom$\DTx(\Gamma) \fCenter \DTx(\Delta), \DTx(A)$
\Axiom$\DTx(\Gamma) \fCenter \DTx(\Delta), \DTx(B)$
\BinaryInf$\DTx(\Gamma) \fCenter \DTx(\Delta), \DTx(A\land B)$
\end{prooftree}
This can be fixed up to be a valid inference using cuts with the
sequents $\DTx(A), \DTx(B) \sequent \AndDT(\DTx(A),\DTx(B))$
and $\AndDT(\DTx(A),\DTx(B)) \sequent \DTx(A)$
and $\AndDT(\DTx(A),\DTx(B)) \sequent \DTx(A)$. These
three sequents have polynomial size proofs by Lemma~\ref{lem:eDTAndOr}.\footnote{%
As stated in the previous footnote, this use of Lemma~\ref{lem:eDTAndOr} is
the reason the $\DT$ proof ends up dag-like instead of tree-like.}

The $\landLeft$, $\lorLeft$ and $\lorRight$ inferences
in~$\pi$ are handled similarly. Other inferences in~$\pi$ are
trivial to handle.

After fixing up the inferences in~$\pi^\prime$ in this way,
we obtain a valid $\DT$ proof~$\pi_1$ of the sequent
$\DTx(\Gamma) \sequent \DTx(\Delta)$ where $\Gamma \sequent \Delta$ is
the final line of~$\pi$.

For polynomial simulation, the last
line of~$\pi$ is a sequent of the form~(\ref{eq:fromLK}),
namely $\Gamma$ is a multiset of disjunctions of literals,
and $\Delta$ is a multiset of conjunctions of literals. Referring
to Equation~(\ref{eq:fromLK}), a conjunct will
$\bigwedge \vec b_i$ will have the conjunctions nested in a balanced fashion
by our assumption that formulas in~$\pi$ have logarithmic height.
However, it is straightforward to give a polynomial size, cut free
$\DT$ proof of $\DTx(\bigwedge \vec b_i)\sequent \Conj( \vec b_i)$ for an arbitrary nesting
of conjunctions in $\bigwedge \vec b_i$.
Likewise, there are polynomial size, cut-free
$\DT$-proofs of $\Disj(\vec a_i) \sequent \DTx(\bigvee \vec a_i)$.
Adding cuts with these to the end of~$\pi_1$ gives the desired
polynomial size $\DT$ proof of (\ref{eq:toDT}).
\end{proof}

\subsection{$\LK$ quasipolynomially simulates $\eLNDT$}\label{sec:LK_ELNDT}

The intuition for the next simulation is that $\eNDT$ formulas
define nondeterministic logspace properties, and these are
expressible with quasipolynomial size Boolean formulas.

\begin{theorem}\label{thm:LK_ELNDT}
$\LK$ quasipolynomially simulates $\eLNDT$. As a result, $\LK$ also quasipolynomially
simulates $\eLDT$.
\end{theorem}

\begin{proof}[Proof sketch]
Suppose $\pi$ is an $\eLNDT$ proof of a sequent
$\Gamma \sequent \Delta$ of $\eNDT$ formulas,
and with associated extension axioms
$\mathcal{A} = \{e_i \liff A_i \}_{i\in I}$. We must
construct an LK proof~$\pi^\prime$
quasipolynomially simulating~$\pi$. The idea for forming~$\pi^\prime$
is to give truth definitions for all formulas
appearing in $\pi$, and then prove that all
sequents are in~$\pi$ are true under these truth
definitions. The truth definition will be based on
st-connectivity in a directed graph~$G_\pi$.  The
nodes of~$G_\pi$ will be the subformulas of
formulas in $\pi$ or~$\cal A$; the edges will be defined
in terms of the literals~$p$ used in~$\pi$. It is well-known
that there are quasipolynomial formulas expressing
st-connectivity in~$G_\pi$. Furthermore, by
\cite{Buss:QuasiPolyPHP}, straightforward constructions
of these quasipolynomial formulas can be used in
$\LK$ proofs to prove basic properties
of st-connectivity.\footnote{The analogous
results were earlier formulated within the bounded
arithmetic theory $\Utheory^1_2$ by
Beckmann-Buss~\cite{BeckmannBuss:localImprove}.
$\Utheory^1_2$ has proof theoretic strength corresponding
to polynomial space, or under the RSUV isomorphism to
quasilogarithmic (that is, $(\log n)^{O(1)}$) space. Likewise,
it corresponds to propositional provability with
$2^{n^{O(1)}}$ size $\LK$ proofs, or under the
RSUV isomorphism, with propositional provability
with polynomial size $\LK$ proofs.  This last claim
does not appear explicitly in the literature, but
see Dowd~\cite{Dowd:PSA,Dowd:PhD} and
Beckmann-Buss~\cite{BeckmannBuss:NPsearchFrege}. }

We describe the direct graph $G_\pi$ in more detail. Consider all
distinct subformulas appearing either (a)~in some
formula~$A$ in~$\pi$ or (b)~in some $A_i$ from
the extension axioms.  These subformulas are
vertices of the graph~$G_\pi$. In addition, $G_\pi$
contains one additional vertex, called~1. For example, suppose
that the formula~$A := (e_1\, \overline p\, p)$ appears
in~$\pi$ and that $e_1 \liff (\overline q \, p \, p)$ is
an extension axiom in~$\cal A$. These contribute the following
nodes to~$G_\pi$:
\begin{equation}\label{eq:GpiExample}
 (e_1\, \overline p\, p), \qquad e_1, \qquad
 p, \qquad  (\overline q \, p \, p), \qquad
 \overline q , \qquad\hbox{and} \qquad 1.
\end{equation}
Enumerate the the vertices
of $G_\pi$ in any arbitrary order as
$v_0, v_1, \dots, v_m$, say with
$v_0$ the vertex~1 and the rest of the vertices
in arbitrary order.  Note $m$ is polynomially bounded (in fact,
linearly bounded) by $|\pi|$.

The edges present in~$G_\pi$
are specified by Boolean formulas~$\varphi_{i,j}$
for distinct $i, j$ in $\{0,\dots,m \}$, so that $\varphi_{i,j}$
is true if there is a directed edge from $v_i$ to~$v_j$
in~$G_\pi$. For a vertex $v_i$ of~$G_\pi$ equal to a formula
$(A p B)$, and let the vertices $v_j$ and $v_{j^\prime}$ in~$G_\pi$
be the $\DT$ formulas $A$ and~$B$. Then $\varphi_{i,j}$ is
the Boolean formula~$\overline p$ and $\varphi_{i,j^\prime}$ is
the Boolean formula~$p$. For vertex~$v_i$ equal to some~$e_k$
and vertex~$v_j$ equal to~$A_k$, then $\varphi_{i,j}$ is
the constant Boolean formula $\top$.  Third, if the vertex~$v_i$
is a DT formula~$p$ with $p$ a literal, then $\varphi_{i,0}$ is
the Boolean formula~$p$. All other formulas $\varphi_{i,j}$
are defined to equal the constant Boolean formula~$\perp$.
(Strictly speaking, $\top$ and~$\perp$ are not allowed constants
for Boolean formulas; instead, they stand for $(p\lor \overline p)$
and $(p\land \overline p)$ for some literal~$p$.)

Returning to the example, let $v_1,\dots, v_5$ be the five
formulas in the order indicated in~(\ref{eq:GpiExample}),
and $v_0$ be~1.
Then, $\varphi_{1,2}$~is~$p$;
$\varphi_{1,3}$~is~$\overline p$;
$\varphi_{2,4}$~is~$\top$;
$\varphi_{3,0}$~is~$p$;
$\varphi_{4,5}$~is~$\overline p$;
$\varphi_{4,3}$~is~$p$;
and $\varphi_{5,0}$~is~$\overline q$.

Finally, for $v_i$ a vertex in~$G_\pi$, namely
a subformula used in~$\pi$, define $\Reach_i$
to be a Boolean formula expressing that there
is a path in $G_\pi$ from $v_i$ to~$v_0$. As discussed
in~\cite{Buss:QuasiPolyPHP}, $\Reach_i$ can
be expressed by a quasipolynomial size formula, and
there are quasipolynomial size proofs of elementary
properties of $\Reach_i$, notably of
\begin{equation}\label{eq:reachProp}
Reach_i ~\liff~ \bigvee\nolimits_{j\not=i} \bigl(
                    \varphi_{i,j} \land \Reach_j\bigr)
\end{equation}
Each line in~$\pi$ is a sequent of the form
\[
v_{i_1},\dots, v_{i_k} \sequent
v_{j_1}, \dots, v_{j_\ell}.
\]
To form the $\LK$ proof~$\pi^\prime$, replace each such
sequent with the quasipolynomial size sequent
\[
\Reach_{i_1},\dots, \Reach_{i_k} \sequent
\Reach_{j_1}, \dots, \Reach_{j_\ell}.
\]
It is now easy to fix up $\pi^\prime$ be a valid
$\LK$ proof.  Initial sequents are handled trivially,
since if $v_i$ is $p$ then $\Reach_i$ is also~$p$.
The only non-trivial inferences are
decision rules \emph{dec-l} and \emph{dec-r}
and these are readily handled with the aid of~(\ref{eq:reachProp}).
\end{proof}

\section{Conclusions}
This work presented sequent-style systems $\LDT$, $\LNDT$, $\eLDT$ and $\eLNDT$ that manipulate decision trees, nondeterministic decision trees, branching programs (via extension) and nondeterministic branching programs (also via extension) respectively.
We examined their relative proof complexity and also compared them to (bounded depth) Frege systems (more precisely their representations in the sequent calculus).

In particular, since (nondeterministic) Branching Programs constitute a natural nonuniform version of (nondeterministic) $\logspace$, the system $\eLDT$ ($\eLNDT $) can be seen as a natural propositional system for (nondeterministic) logspace.
This mimics the way that $\LK$ (or the Frege system) is a natural system for $\Alogtime$ (via Boolean formulas) and $\eLK$ (or extended Frege) is a natural system for $\Ptime$ (via Boolean circuits).
\medskip

We did not compare the proof complexity theoretic strength of our systems $\eLDT $ and $\eLNDT$ with the system for $\logspace$ in \cite{Cook:EdinburghSlides} and the systems for $\logspace$ and $\NL$ in \cite{Perron:logspace,Perron:thesis}.
In future work we intend to show that our systems correspond to the bounded arithmetic theories $\VL$ and $\VNL$, in the usual way.
Namely, proofs of $\Pi_1 $ formulas in $\VL $ translate to families of small $\eLDT$ proofs of each instance, and, conversely, $\VL$ proves the soundness of $\eLDT$. Similarly for $\VNL$ and $\eLNDT$.
This would render our systems polynomially equivalent to their respective systems from \cite{Cook:EdinburghSlides,Perron:logspace,Perron:thesis}, though this remains work in progress.

There are two natural open questions arising from this work.
The first concerns the exact relationship between $\LDT$ and low-depth systems:

\begin{question}
	\label{quest:tree-1lk-vs-tree-ldt}
	Does tree-$\oneLK$ polynomially simulate tree-$\LDT$, or is there a quasipolynomial separation between the two?
\end{question}

The second open question is whether tree-like systems for branching programs may polynomially simulate their corresponding dag-like ones.

\begin{question}
	\label{quest:tree-eldt-vs-eldt}
	Does tree-$\eLDT$ polynomially simulate $\eLDT$? Similarly for $\eLNDT$
\end{question}

While well-defined, the systems tree-$\eLDT$ and tree-$\eLNDT$ do not seem very robust, in the sense that it is not immediate how to witness branching program isomorphisms with short proofs, cf.~\ref{lem:extNamesEquiv}.
Nonetheless, it would be interesting to settle their proof complexity theoretic status.

There has been much recent work on the proof complexity of systems that may manipulate OBDDs \cite{Knop:IPSlike,BussIKS18,ItsyksonKRS17}, a special kind of branching program where propositional variables must occur in the same relative order on each path through the dag.
In fact, we could also define an `OBDD fragment' of $\eLDT$ by restricting lines to $\eDT$ formulas expressing OBDDs, as alluded to in Example~\ref{example:bp-to-edt-formula-th42}.
It would be interesting to examine such systems from the point of view of proof complexity in the future, in particular comparing them to existing OBDD systems.

In this work we restricted the expressivity of all lines in a proof in order to define our various systems.
An alternative approach is to restrict only the cut-formulas.
Over conclusions of the appropriate form, this makes no difference to the notion of a proof thanks to the subformula property, but such systems have the advantage of being complete for all classes of formulas (for instance, via cut-free completeness).
In this way we could have rather considered one single ambient system consisting of the connectives and rules for decision literals, disjunction and conjunction.
Our various systems could thence be recovered by only restricting cut formulas.
Many of our results already go through in this setting with respect to the provability of arbitrary formulas.

\bibliographystyle{siam}
\bibliography{localPropDTandBP,logic}
\end{document}